\documentclass[prb,twocolumn,showpacs,nofootinbib,preprintnumbers,amsmath,amssymb,aps,longbibliography,superscriptaddress]{revtex4-2}
\usepackage[T1]{fontenc}
\usepackage[utf8]{inputenc}
\usepackage{subfigure, lmodern, amsmath,amssymb, graphicx, pifont, adjustbox, bm, xcolor}
\usepackage{amsfonts}
\usepackage{amsthm}
\usepackage{comment}
\usepackage{mathtools}
\usepackage{float}
\usepackage{braket}
\usepackage{longtable}
\usepackage{graphicx, subfigure}
\usepackage{xcolor}
\usepackage{tikz}
\usepackage{enumitem}
\usepackage{dsfont}
\usepackage{soul}

\newtheorem{theorem}{Theorem}
\newtheorem{definition}{Definition}
\newtheorem{proposition}{Proposition}
\newtheorem{corollary}{Corollary}

\DeclarePairedDelimiter\floor{\lfloor}{\rfloor}
\makeatletter\g@addto@macro\bfseries{\boldmath}\makeatother

\usepackage[colorlinks=true, citecolor=blue, urlcolor=blue, linkcolor=blue, breaklinks=true, pdfpagelabels=false]{hyperref}

\newcommand{\appendixref}[1]{\hyperref[#1]{appendix~\ref{#1}}}
\def\equationautorefname~#1\null{eq.\,(#1)\null}

\newcommand{\ie}{\begin{equation}\begin{aligned}}
\newcommand{\fe}{\end{aligned}\end{equation}}

\newcommand{\lket}[1]{\langle #1 |}
\newcommand{\rket}[1]{|#1\rangle}
\newcommand{\lrket}[2]{\langle #1 |#2\rangle}
\newcommand{\projection}{\hat{\mathcal{P}}}
\newcommand{\hamiltonian}{\hat{H}}

\usepackage{cleveref}
\crefname{appendix}{App.}{Apps.}
\crefname{equation}{Eq.}{Eqs.}
\crefname{figure}{Fig.}{Figs.}
\crefname{table}{Tab.}{Tabs.}
\crefname{section}{Sec.}{Secs.}
\crefname{definition}{Definition}{Definition}
\crefname{lemma}{Lemma}{Lemma}
\crefname{theorem}{Theorem}{Theorem}
\crefname{proposition}{Proposition}{Proposition}
\crefname{corollary}{Corollary}{Corollary}

\newcommand{\tket}[1]{\ket{\;\tikz[baseline={([yshift=-.5ex]current bounding box.center)}]{#1}\;}}
\definecolor{mydarkgreen}{rgb}{0.0, 0.5, 0.0}

\newcommand{\bs}{\textcolor{blue}{\bullet}}
\newcommand{\gs}{\textcolor{mydarkgreen}{\bullet}}
\newcommand{\rs}{\textcolor{red}{\bullet}}

\begin{document}

\title{Bridging Commutant and Polynomial Methods for Hilbert Space Fragmentation}
\date{\today}

\author{Bo-Ting Chen}
\thanks{These authors contributed equally to this work.}
\affiliation{Department of Physics, Princeton University, Princeton, New Jersey 08544, USA}
\author{Yu-Ping Wang}
\thanks{These authors contributed equally to this work.}
\affiliation{Department of Physics, Stony Brook University, Stony Brook, New York 11790, USA}
\author{Biao Lian}
\affiliation{Department of Physics, Princeton University, Princeton, New Jersey 08544, USA}

\begin{abstract}
A quantum model exhibits Hilbert space fragmentation (HSF) if its Hilbert space decomposes into exponentially many dynamically disconnected subspaces, known as Krylov subspaces. A model may however have different HSFs depending on the method for identifying them. Here we establish a connection between two vastly distinct methods recently proposed for identifying HSF: the commutant algebra (CA) method and integer characteristic polynomial factorization (ICPF) method. For a Hamiltonian consisting of operators admitting rational number matrix representations, we prove a theorem that, if its center of commutant algebra have all eigenvalues being rational, the HSF from the ICPF method must be equal to or finer than that from the CA method. We show that this condition is satisfied by most known models exhibiting HSF, for which we demonstrate the validity of our theorem. We further discuss representative models for which ICPF and CA methods yield different HSFs. Our results may facilitate the exploration of a unified definition of HSF. 
\end{abstract}

\maketitle

Thermalization of out-of-equilibrium closed quantum systems has attracted extensive interests of study in both theory and experiment in recent years. While most non-integrable quantum many-body systems obey the eigenstate thermalization hypothesis (ETH) \cite{deutsch1991quantum,srednicki1994chaos,rigol2008thermalization, polkovnikov2011colloquium} and thermalize quickly, there are notable exceptions. For instance, Many-body localized systems \cite{anderson1958,gornyi2005,baa2006,rahul2015review,altman2015review,abanin2018review} are prevented from thermalization by effective quasi-local conserved quantities. Another class of examples are models with quantum many-body scars \cite{bernien2017probing,turner2017quantum, moudgalya2018exact, choi2018emergent,ho2018periodic,moudgalya2018nonint,serbyn_review,chandran_review}, which are a measure zero set of eigenstates violating ETH and can give rise to non-thermalizing dynamics when they significantly overlap with the initial state. Recently, Hilbert space fragmentation (HSF) \cite{Sala_2020,khemani2020hilbert,rakovszky2019statistical,moudgalya2019thermalization,sanjay_review} has been extensively studied as a new mechanism for breaking ergodicity and thermalization. A quantum many-body model with system volume $L$ exhibits HSF if its Hilbert space decomposes into $\exp(aL)$ ($a>0$) number of dynamically disconnected subspaces known as \emph{Krylov subspaces}, which lead to non-ergodic dynamics. 

Importantly, a model may have different HSFs depending on how they are identified. Two distinct methods for identifying HSF proposed recently are the commutant algebra (CA) method~\cite{moudgalya2021hilbert,Moudgalya2023_numerical}, and the integer characteristic polynomial factorization (ICPF) method~\cite{Regnault:2022ocy}. The CA method deals with families of Hamiltonians, which elucidates how Krylov subspaces are analogous to and distinct from conventional symmetry sectors, and can be extended to open quantum systems \cite{Li2023HSFopen}. However, it is computationally highly expensive, and only a few models have been solved analytically \cite{moudgalya2021hilbert,Moudgalya2024_QMBS_CA,Yves2025minimalHubbard}. The ICPF method deals with Hamiltonians admitting rational number matrix representations (denoted as \emph{rational operators} hereafter), which is numerically significantly more efficient, but its physical meaning is unclear. 

Interestingly, we find the CA and ICPF methods yield the same HSF for many known models, including the pair-flip model~\cite{caha2018pairflip}, Temperley-Lieb models~\cite{PhysRevB.40.4621, Aufgebauer_2010, M.T.Batchelor_1990, Read2007ti}, $t$-$J_z$ model~\cite{PhysRevB.55.6491, PhysRevLett.85.4755} and dipole-conserving model~\cite{Sala_2020,khemani2020hilbert,Moudgalya_2021}, etc (see SM \cite{SM}). This motivates us to investigate the relation between these two vastly different methods. In this letter, we prove a theorem for Hamiltonians from linear combinations of rational operators: under the condition that all the rational operators in the center of their commutant algebra have rational eigenvalues, HSF from the ICPF method must be equal to or finer than HSF from the CA method. This condition is satisfied by all the aforementioned models. We also prove that if ICPF yields a finer HSF than CA, its Krylov subspace basis must depend on Hamiltonian coefficients. We verify our conclusions in the above models, as well as in additional representative models with distinct HSFs from the two methods.

\emph{Setup}. We study quantum systems in a lattice with $L$ sites, where each site has a finite-dimensional Hilbert space. The Hamiltonian takes the generic form:
\ie
\label{eqn: Hamiltonian in J}
\hamiltonian(\{J_\eta\})\equiv \sum_{\eta=1}^{N}J_{\eta} \hat{h}_{\eta}
\ ,\fe
where $\hat{h}_\eta$ are hermitian terms supported locally on finite numbers of sites, $J_{\eta}$ are real coefficients, and $N$ is the total number of terms ($N\sim \mathcal{O}(L)$ for a local Hamiltonian). Importantly, in order for the ICPF method \cite{Regnault:2022ocy} of HSF to apply, we assume there exists an \emph{orthonormal} basis $\ket{u_i}$, denoted as the \emph{Hamiltonian basis}, in which every operator $\hat{h}_\eta$ is a rational operator. Here we define an operator $\hat{h}$ as a \emph{rational operator} if all of its matrix elements in the Hamiltonian basis are rational numbers, i.e., $\langle u_i|\hat{h}|u_j\rangle \in \mathbb{Q}(i)=\mathbb{Q}+i\mathbb{Q}$. This assumption, while seemingly restrictive, is satisfied by most quantum models exhibiting HSF, and the Hamiltonian basis is usually a product state basis.

A quantum model is said to exhibit HSF if there exists a Hilbert space basis in which the Hamiltonian takes a block diagonal form $\hat{H}=\bigoplus_{k=1}^{N_K}K_k(\hat{H})$, and the number of blocks $N_K$ grows exponentially as $\exp(aL)$ with some $a>0$. Each block $K_k(\hat{H})$ is a $D_k\times D_k$ matrix ($D_k\ge1$) acting on a $D_k$-dimensional Hilbert space $\mathcal{H}_k$, which is called a \emph{Krylov subspace}. We can define one of its basis vector as the root state $|\Psi_k\rangle$, thus the Krylov subspace $\mathcal{H}_k$ is spanned by $D_k$ linearly independent state vectors $\{\ket{\Psi_k}, \hamiltonian\ket{\Psi_k}, \cdots, \hamiltonian^{D_k-1}\ket{\Psi_k}\}$. By definition, $\hamiltonian^{D_k}\ket{\Psi_k}$ will be linearly dependent on the above $D_k$ vectors, thus the Krylov subspace is closed under the action of $\hamiltonian$.

The HSF of a model is basis dependent and may not be unique. Recently, two inequivalent methods, the ICPF method \cite{Regnault:2022ocy} and the CA method \cite{moudgalya2021hilbert,Moudgalya2023_numerical}, are proposed for identifying the HSF. The aim of this work is to establish the connection between them. In the below, we first briefly review the two methods.

\emph{The ICPF method}. 
This method \cite{Regnault:2022ocy} identifies the HSF by taking certain rational coefficients (of interest) $J_\eta\in\mathbb{Q}$ in the Hamiltonian $\hat{H}$ of \cref{eqn: Hamiltonian in J}, and then factorizing its characteristic polynomial into irreducible rational polynomials: 
\ie
\label{eqn: ICPF factorization form}
\Phi(\mu)\equiv \det\big(\mu \mathds{1} - \hamiltonian(\{J_{\eta}\})\big)
= \prod_{p} \left[ \phi_{p}^{(D_p)}(\mu) \right]^{d_{p}},
\fe
where $D_p$ and $d_p$ are positive integers. Each factor $\phi_{p}^{(D_p)}(\mu)$ is a degree $D_p$ polynomial in $\mu$ with all the coefficients being rational numbers, thus called a \emph{rational polynomial}, and $d_p$ is its multiplicity. Each $\phi_{p}^{(D_p)}(\mu)$ is irreducible (not rationally factorizable). Such a factorization exists, since the Hamiltonian $\hat{H}$ and thus the characteristic polynomial $\Phi(\mu)$ have rational coefficients in the Hamiltonian basis. Note that one can always turn $\Phi(\mu)$ into an integer coefficient polynomial $\widetilde{\Phi}(\mu)=q\Phi(\mu)$ by multiplying some integer $q$. As shown in \cite{Regnault:2022ocy}, the rational factorization of \cref{eqn: ICPF factorization form} is unique and equivalent to factorizing $\widetilde{\Phi}(\mu\}$ into irreducible integer polynomials (thus the name ICPF), by Gauss's lemma \cite{atiyah2018introduction}.

As proved in \cite{Regnault:2022ocy}, \cref{eqn: ICPF factorization form} implies a fragmentation of the full Hilbert space $\mathcal{H}$ into a direct sum
\ie
\label{eqn: Hdecomp-ICPF}
\mathcal{H} = \bigoplus_{p}\Big[\bigoplus_{\delta_p=1}^{d_p}\mathcal{H}_{p,\delta_p}\Big]\ ,
\fe
where each $\mathcal{H}_{p,\delta_p}$ is a dimension $D_p$ Krylov subspace that one-to-one corresponds to a polynomial factor $\phi_{p}^{(D_p)}(\mu)$, and $1\le \delta_p\le d_p$ labels $d_p$ degenerate Krylov subspaces (which have identical matrix representations of the Hamiltonian $\hat{H}$). Importantly, the root state $|\Psi_{p,\delta_p}\rangle$ and hence the whole basis (generated by $\hamiltonian$ from $|\Psi_{p,\delta_p}\rangle$) of each Krylov subspace $\mathcal{H}_p^{(\delta_p)}$ are \emph{rational states}. Here a \emph{rational state} $|\Psi\rangle=\sum_{j}\psi_j|u_j\rangle$ is defined to have rational ratios $\psi_i/\psi_j\in\mathbb{Q}(i)$ between any two nonzero components in the Hamiltonian basis $|u_j\rangle$.

\emph{The CA method}.
This approach identifies the HSF based on the commutant algebra of the Hamiltonian terms $\{ \hat{h}_{\eta} \}$ in \cref{eqn: Hamiltonian in J}. Let $\mathcal{A}$ be the von Neumann algebra generated by $\{ \hat{h}_{\eta} \}$, which is the set of operators given by additions and multiplications among $\{\hat{h}_{\eta}\}$ and any complex numbers. Let $\mathcal{C}$ be its commutant algebra defined as the set of operators commuting with all terms $\{ \hat{h}_{\eta} \}$. The two mutually commuting algebras $\mathcal{A}$ and $\mathcal{C}$ allow a decomposition of the Hilbert space $\mathcal{H}$ into irreducible representations (irreps) of $\mathcal{A} \times \mathcal{C}$\cite{Harlow:2016vwg}, leading to an HSF into a direct sum of virtual bipartitions:
\begin{equation}
	\label{eqn: Hdecomp}
	\mathcal{H} = \bigoplus_{\lambda}
	\left[ \mathcal{H}_{\lambda}^{(\mathcal{A})}
		\otimes
		\mathcal{H}_{\lambda}^{(\mathcal{C})} \right],
\end{equation}
where $\mathcal{H}_{\lambda}^{(\mathcal{A})}$ is a $D_\lambda$ dimensional irrep of $\mathcal{A}$, and $\mathcal{H}_{\lambda}^{(\mathcal{C})}$ is a $d_\lambda$ dimensional irrep of $\mathcal{C}$. Particularly, if $\mathcal{C}$ is an Abelian algebra, all $d_\lambda\equiv 1$. Each bipartition sector is labeled by a distinct summand index $\lambda$. In the basis of \cref{eqn: Hdecomp}, an operator in $\mathcal{A}$ (e.g. the Hamiltonian $\hamiltonian$) has a block diagonal matrix form $\hat{h}_\mathcal{A} =\bigoplus_{\lambda}\big[ M_{\lambda}(\hat{h}_\mathcal{A})\otimes\mathds{1}_{d_{\lambda}} \big]$, while an operator in $\mathcal{C}$ takes the form $\hat{h}_\mathcal{C} =\bigoplus_{\lambda}\big[ \mathds{1}_{D_{\lambda}}\otimes N_{\lambda}(\hat{h}_\mathcal{C}) \big]$. 
Thus, the Hilbert space of each summand $\lambda$ decomposes into $d_\lambda$ degenerate Krylov subspaces $\mathcal{H}_{\lambda,\delta_\lambda}$ ($1\le\delta_\lambda\le d_\lambda$), each of dimension $D_\lambda$:
\ie
\label{eqn: HCA-decomp}
\mathcal{H}_{\lambda}^{(\mathcal{A})}\otimes\mathcal{H}_{\lambda}^{(\mathcal{C})}=\bigoplus_{\delta_\lambda=1}^{d_\lambda}\mathcal{H}_{\lambda,\delta_\lambda}\ .
\fe
For later purpose, we denote an orthogonal (not necessarily normalized) basis $\ket{\lambda, \delta_\lambda, \Delta_\lambda}$ ($\Delta_\lambda=1,\cdots,D_\lambda$) for each Krylov subspace $\mathcal{H}_{\lambda,\delta_\lambda}$ as a \emph{CA Krylov basis}.

\emph{Bridging ICPF with CA.} The Krylov subspaces identified by ICPF and CA can be generically different. To reveal the connection between them, we define a partial-order relation as follows: 
\begin{itemize}
\item If for any Krylov subspace $\mathcal{H}_{p,\delta_p}$ from ICPF (\cref{eqn: Hdecomp-ICPF}), there exists a summand $\lambda$ and an index $\delta_\lambda$ in some CA Krylov basis (\cref{eqn: HCA-decomp}) such that $\mathcal{H}_{p,\delta_p}\subset \mathcal{H}_{\lambda,\delta_\lambda}$, we say ICPF $\succeq$ CA, meaning that ICPF gives an equal or finer HSF. 
\item If there exists a CA Krylov basis such that for any $\lambda$ and $\delta_\lambda$, there is some $p$ and $\delta_p$ such that $\mathcal{H}_{\lambda,\delta_\lambda}\subset\mathcal{H}_{p,\delta_p}$, we say ICPF $\preceq$ CA (CA gives an equal or finer HSF).
\end{itemize}
If both of the above relations hold, we have ICPF $=$ CA and their HSFs are identical. If only one of the relation holds, the partial order is strict, and one has ICPF $\succ$ CA or ICPF $\prec$ CA. 

Specifically, the relation ICPF $\succeq$ CA is equivalent to the existence a CA Krylov basis where all basis states $\ket{\lambda, \delta_\lambda, \Delta_\lambda}$ are rational states. This is because the HSF identified by ICPF necessarily has a rational basis of Krylov subspaces. If one further has identical dimension $D_\lambda=D_p$ in each block, one has ICPF $=$ CA, as ensured by uniqueness of the factorization in \cref{eqn: ICPF factorization form}.

However, explicitly constructing the CA Krylov basis for a given model could be rather involved. In the below, we show that it is sufficient to tell whether ICPF $\succeq$ CA simply from the algebraic structure of $\mathcal{A}$, and there is no need to construct the Krylov basis. 

To this end, for the two commuting algebras $\cal A$ and $\cal C$ defined earlier, we define the center ${\cal Z } = {\cal A }\cap {\cal C}$, which is a self-commuting algebra. In the CA basis of \cref{eqn: Hdecomp}, $\cal Z$ consists of all operators of the form $\hat{h}_\mathcal{Z} =
\bigoplus_{\lambda}
\left[ 
c_{\lambda}(\hat{h}_\mathcal{Z})
\mathds{1}_{D_{\lambda}}
\otimes
\mathds{1}_{d_{\lambda}} 
\right]
$, where $c_{\lambda}(\hat{h}_\mathcal{Z})$ is some eigenvalue of $\hat{h}_\mathcal{Z}$. The center $\mathcal{Z}$ thus serves as a handy tool for identifying distinct summands $\lambda$. We define the basis of the center $\cal Z$, denoted by $\bold{Basis}(\mathcal{Z})=\{\hat{z}_k\}$, as a minimal set of \emph{rational operators} $\hat{z}_k\in \cal Z$ that generates $\cal Z$. 
The existence of such a set of rational operators $\hat{z}_k$ is ensured by the fact that both $\mathcal{A}$ and $\mathcal{C}$ are generated by rational operators. We then prove the following theorem:

\begin{theorem}[\textbf{The main theorem}]
\label{theorem: main theorem}
	Assume $\mathcal{Z}$ is the center of commutant algebra of a quantum model. If every element $\hat{z}_k\in \bold{Basis}(\mathcal{Z})$ (which is a rational operator) has all its eigenvalues $c(\hat{z}_k)$ being rational (i.e. $c(\hat{z}_k)\in \mathbb{Q}(i)$), the model has ICPF $\succeq$ CA.
\end{theorem}

\emph{Proof sketch.} Here we sketch the proof of \cref{theorem: main theorem}, the complete details of which can be found in the Supplemental Material (SM \cite{SM}). The key idea is to notice that in the CA Krylov basis of \cref{eqn: Hdecomp}, the Hamiltonian $\hamiltonian$ is block diagonal in $\lambda$, and thus the characteristic polynomial $\Phi(\mu)$ can be written as
\begin{equation}
\label{eqn: CA polynomial}
\Phi(\mu) 
= \prod_\lambda \Phi_\lambda(\mu)
\ , \quad 
\Phi_\lambda(\mu)
= 
\mathrm{det}_\lambda(\mu \mathds{1} - \projection_{\lambda} \hamiltonian \projection_{\lambda})
\ ,\end{equation}
where $\projection_{\lambda}$ denotes the projection operator into the $D_\lambda d_\lambda$ dimensional Hilbert space of each summand $\lambda$ in \cref{eqn: Hdecomp}, and $\Phi_\lambda(\mu)$ is the characteristic polynomial of the Hamiltonian in the block of summand $\lambda$. Explicitly, in terms of the CA Krylov basis, 
\begin{equation}
\projection_\lambda=\mathds{1}_{D_{\lambda}}
\otimes
\mathds{1}_{d_{\lambda}} =\sum_{\delta_\lambda=1}^{d_\lambda}\sum_{\Delta_\lambda=1}^{D_\lambda} \frac{\rket{\lambda, \delta_\lambda, \Delta_\lambda} \lket{\lambda, \delta_\lambda, \Delta_\lambda}}{\lrket{\lambda, \delta_\lambda, \Delta_\lambda}{\lambda, \delta_\lambda, \Delta_\lambda}}
\ .\end{equation}

Assume all the coefficients $J_\eta$ in \cref{eqn: Hamiltonian in J} are rational, hence the Hamiltonian $\hat{H}$ is a rational operator. Accordingly, as we will prove later, each $\projection_\lambda$ is also a rational operator. Thus, each polynomial $\Phi_\lambda(\mu)$ in \cref{eqn: CA polynomial} is by definition a rational polynomial, and $\Phi(\mu) = \prod_\lambda \Phi_\lambda(\mu)$ is readily a (reducible) rational factorization of $\Phi(\mu)$. Since each $\lambda$ summand consists of $d_\lambda$ degenerate Krylov subspaces $\mathcal{H}_{\lambda,\delta_\lambda}$, its characteristic polynomial can be further factorized into $\Phi_\lambda(\mu) = \big[ \phi_{\lambda}^{(D_\lambda)}(\mu) \big]^{d_{\lambda}}$, where $\phi_{\lambda}^{(D_\lambda)}$ (independent of $\delta_\lambda$) is the degree $D_\lambda$ characteristic polynomial of each Krylov subspace $\mathcal{H}_{\lambda,\delta_\lambda}$. Moreover, the rational nature of $\Phi_\lambda(\mu)$ ensures that $\phi_{\lambda}^{(D_\lambda)}(\mu)$ is also rational, which can be proved by explicit expansion, as well as by  Galois theory (see SM \cite{SM} for details). Thus, we arrive at a (reducible) rational factorization of the characteristic polynomial
$\Phi(\mu) = \prod_\lambda \big[ \phi_{\lambda}^{(D_\lambda)}(\mu) \big]^{d_{\lambda}}$. By uniqueness of the irreducible rational factorization, every $\phi_{p}^{(D_p)}(\mu)$ from ICPF in \cref{eqn: ICPF factorization form} must be a factor of some polynomial $\phi_{\lambda}^{(D_\lambda)}(\mu)$, which ensures the corresponding Krylov subspace $\mathcal{H}_{p,\delta_p}\subset \mathcal{H}_{\lambda,\delta_\lambda}$. Thus, ICPF $\succeq$ CA.

We now prove that every $\projection_\lambda$ is a rational operator. First, notice that $\projection_\lambda\in\mathcal{Z}$ is an element of the center, thus is generated by $\bold{Basis}(\mathcal{Z})=\{\hat{z}_k\}$. Assume $\hat{z}_k$ takes eigenvalue $c_\lambda(\hat{z}_k)$ in summand $\lambda$. The completeness of $\bold{Basis}(\mathcal{Z})$ then ensures $\projection_\lambda$ can be expressed as 
\begin{equation}
\label{eqn: explicit expression of P lambda k}
    \projection_{\lambda}
    =
    \prod_{k}
    \left(
    \prod_{\lambda^{\prime}\in \Lambda(\lambda,k)}
    \frac
    {\hat{z}_k - c_{\lambda^{\prime}}(\hat{z}_k)}
    {c_{\lambda}(\hat{z}_k) - c_{\lambda^{\prime}}(\hat{z}_k)}
    \right)
,\end{equation}
where we have defined $\Lambda(\lambda,k)=\{\lambda'\ |\  c_{\lambda^{\prime}}(\hat{z}_k) \neq c_{\lambda}(\hat{z}_k)\}$ to be the set of summands $\lambda'$ with $\hat{z}_k$ eigenvalue different from that of $\lambda$. Therefore, if all the eigenvalues of $\hat{z}_k$ are rational, every $c_\lambda(\hat{z}_k)\in \mathbb{Q}(i)$ is rational, and thus every $\projection_{\lambda}$ as expressed in \cref{eqn: explicit expression of P lambda k} is a rational operator. \qed

When the Hamiltonian basis $\{|u_i\rangle\}$ is a product state basis, HSF can be classified into two types: classical fragmentation (cHSF) and quantum fragmentation (qHSF). If every Krylov subspace is spanned by a subset of the Hamiltonian basis $\{|u_i\rangle\}$, the resulting HSF is defined as cHSF; otherwise, the HSF is defined as qHSF. Whether a model has cHSF or qHSF may depend on the HSF method, e.g., ICPF or CA (see examples in \cref{table: various models}). We then have the following corollary of \cref{theorem: main theorem}:

\begin{corollary}
\label{corollary1}
If a model has cHSF under the CA method, the model has ICPF $\succeq$ CA. 
\end{corollary}

This is because in this case, the CA Krylov basis coincides with the Hamiltonian basis, making $\hat{P}_{\lambda}$ readily diagonal with rational eigenvalues either 0 or 1.

Additionally, we have the following proposition regarding the $J_\eta$ dependence of HSF:

\begin{proposition}
\label{proposition: when ICPF > CA}
If a model has ICPF $\succ$ CA, namely, there exist certain CA Krylov subspaces that admit a finer fragmentation $\mathcal{H}_{\lambda,\delta_\lambda}=\bigoplus_{j} \mathcal{H}_{p_j^{(\lambda)},\delta_{p_j}^{(\lambda)}}$ by ICPF, then at least one of the ICPF Krylov subspaces $\mathcal{H}_{p_j^{(\lambda)},\delta_{p_j}^{(\lambda)}}$ must have a basis that depends on the coefficients $\{J_\eta\}$.
\end{proposition}

Intuitively, \cref{proposition: when ICPF > CA} follows from the fact that the CA method is independent of the coefficients $\{J_{\eta}\}$. Thus, any finer fragmentation by ICPF must result from additional structures brought by the Hamiltonian coefficients $\{J_{\eta}\}$. For a detailed proof, see SM \cite{SM}.

\begin{table*}[htbp]
\centering
\caption{The HSF of four representative models, which show agreement with \cref{theorem: main theorem}, \cref{corollary1} and \cref{proposition: when ICPF > CA}.
}
\label{table: various models}
\begin{tabular}{|l||c|c|c|c|c|c|c|}
\hline
\textbf{Model} & \textbf{Hamiltonian} & \textbf{Commutant} $\mathcal{C}$ & $\mathbf{Basis}(\mathcal{Z})$ \textbf{eigenvalues} & \textbf{HSF order} & \textbf{ICPF basis} & \textbf{ICPF} & \textbf{CA} \\ 
\hline
Pair-Flip & \cref{eqn: Hamiltonian of PF} & Abelian & Rational & ICPF $=$ CA & $\{J_\eta\}$ independent & cHSF & cHSF
\\
\hline
Temperley-Lieb & \cref{eqn: Hamiltonian of TL} & Non-Abelian & Rational & ICPF $=$ CA & $\{J_\eta\}$ independent & qHSF & qHSF
\\
\hline
Quantum Breakdown & \cref{eqn: Hamiltonian of QB} & Abelian & Rational & ICPF $\succ$ CA & $\{J_\eta\}$ dependent & qHSF & cHSF
\\
\hline
Fibonacci matrix & \cref{eqn: Hamiltonian of FM} & Abelian & Irrational & ICPF $\prec$ CA & $\{J_\eta\}$ independent & cHSF & qHSF
\\
\hline
$t$-$J_z$ & See SM \cite{SM} & Abelian & Rational & ICPF $=$ CA & $\{J_\eta\}$ independent & cHSF & cHSF
\\
\hline
\end{tabular}
\end{table*}

\emph{Examples}. To verify the above results, we numerically examine the HSF in examples of the pair-flip (PF) model \cite{caha2018pairflip}, the Temperley-Lieb (TL) model \cite{PhysRevB.40.4621, Aufgebauer_2010, M.T.Batchelor_1990, Read2007ti}, the quantum breakdown (QB) model \cite{Lian:2022QBM, Chen2024qfEQBM,Hu_2024,Liu2025QBM,hu2025glass}, the Fibonacci matrix (FM) model and the $t$-$J_z$ model \cite{PhysRevB.55.6491, PhysRevLett.85.4755}. In all models, the Hamiltonian basis is the on-site product state basis. 

The HSF structures of these models are summarized in \cref{table: various models}. In the PF, TL, QB and $t$-$J_z$ models, all the center bases in $\bold{Basis}(\mathcal{Z})$ have rational eigenvalues, thus ICPF $\succeq$ CA by \cref{theorem: main theorem}. In particular, the PF, TL and $t$-$J_z$ models have ICPF $=$ CA, and all their Krylov subspace basis are independent of $\{J_\eta\}$. The QB model has strictly ICPF $\succ$ CA, and the ICPF Krylov basis are $\{J_\eta\}$ dependent, in agreement with \cref{proposition: when ICPF > CA}. The FM model provides an example where the elements in $\bold{Basis}(\mathcal{Z})$ have irrational eigenvalues, which has strictly ICPF $\prec$ CA.

In the below, we analyze the HSF of four representative models more specifically, with further details given in the SM \cite{SM}. The ICPF is calculated with random integer $\{J_\eta\}$ to avoid accidental factorizations. A useful necessary and sufficient condition for ICPF $=$ CA is, for each ICPF Krylov subspace characterized by polynomial $\phi_{p}^{(D_p)}(\mu)$, there is a CA Krylov subspace with equal dimension $D_\lambda=D_p$, which has a CA root state $|\Psi\rangle$ satisfying $\phi_{p}^{(D_p)}(\hamiltonian)\ket{\Psi} = 0$.

\emph{(i) The pair-flip (PF) model}, the Hamiltonian of which is defined in a one-dimensional (1D) lattice with $m$ basis states $|\alpha\rangle_i$ ($1\le \alpha\le m$) per site $i$, and parameters $\{g_{i}^{\alpha, \beta},\mu_i^\alpha\}$: 
\begin{equation}
\label{eqn: Hamiltonian of PF}
    H_{\text{PF}} = 
    \sum_{i=1}^{L-1} \sum_{\alpha,\beta=1}^{m} g_{i}^{\alpha, \beta} \hat{F}_{i,i+1}^{\alpha,\beta}
    +
    \sum_{i=1}^{L} \sum_{\alpha=1}^{m} 
    \mu_{i}^{\alpha}
    \hat{N}_{i}^{\alpha}\ ,
\end{equation}
where $\hat{F}_{i,i+1}^{\alpha,\beta} = \left( \rket{\alpha \alpha} \lket{\beta \beta} \right)_{i, i+1}$ are pair-flip terms on sites $i$ and $i+1$, $\hat{N}_{i}^{\alpha} = \rket{\alpha} \lket{\alpha}_i$, and $g_{i}^{\alpha, \beta}={g_{i}^{\beta, \alpha}}^*$. 
Its commutant algebra ${\cal C}$ is Abelian, thus $d_\lambda\equiv1$ in \cref{eqn: Hdecomp}. 
We numerically verified up to $L=8$ that the model has ICPF $=$ CA, giving cHSF. For example, at $L=6$, both ICPF and CA give $D_\lambda=87,47,11,1$ dimensional Krylov subspaces, the numbers of which are $1,6,24,96$. The $D_\lambda=87$ subspace has a root state $\ket{111111}$ satisfying $\phi_{p}^{(87)}(\hamiltonian)\ket{111111} = 0$ for any $\{g_{i}^{\alpha, \beta},\mu_i^\alpha\}$, where $\phi_{p}^{(87)}$ is its polynomial from ICPF. The other subspaces have similar results.

\emph{(ii) The Temperley-Lieb (TL) model}. This model is also defined in a 1D lattice with $m$ basis states $|\alpha\rangle_i$ per site $i$, which has a Hamiltonian with parameters $\{J_{i, i+1}\}$:
\begin{equation}
\label{eqn: Hamiltonian of TL}
H_{\text{TL}} =
\sum_{i=1}^{L-1} J_{i, i+1}
\sum_{\alpha, \beta=1}^{m} 
\left( \ket{\alpha \alpha} \lket{\beta \beta} \right)_{i, i+1}\ .
\end{equation}
This model has a non-Abelian commutant algebra $\mathcal{C}$ and thus $d_\lambda$ not always $1$. We verified up to $L=8$ that the model has ICPF $=$ CA, giving qHSF. For $L=6$, both ICPF and CA give Krylov subspaces with $(D_\lambda,d_\lambda)=(9,8),(5,55),(5,1)$ and $(1,377)$, and their basis coincide. 
For example, the $(D_\lambda, d_\lambda) = (5, 1)$ Krylov subspace has a rational root state $\ket{\Psi}=\ket{\psi_\text{dimer}}_{1,2}\otimes \ket{\psi_\text{dimer}}_{3,4} \otimes \ket{\psi_\text{dimer}}_{5,6}$, with $\ket{\psi_\text{dimer}}_{i,i+1}=\sum_{\alpha}\ket{\alpha\alpha}_{i,i+1}$. It satisfies $\phi_{p}^{(5)}(\hamiltonian)\ket{\Psi}=0$ for any $\{J_{i,i+1}\}$, where $\phi_{p}^{(5)}$ is its degree-5 polynomial from ICPF.

\emph{(iii) The quantum breakdown (QB) model,} which is particle number non-conserving \cite{Lian:2022QBM, Chen2024qfEQBM,Hu_2024,Liu2025QBM,hu2025glass}. We consider a quantum breakdown model with two hardcore boson modes $b_{i,\nu}$ ($\nu=1,2$) per site $i$, and parameters $\{J_i^\nu,\mu_i\}$:
\begin{equation}
\label{eqn: Hamiltonian of QB}
    H_{\text{QB}} =  
    \sum_{i=1}^{L-1} \sum_{\nu=1}^2 J_{i}^{\nu} (b_{i+1,1}^{\dagger} b_{i+1,2}^{\dagger} b_{i,\nu} + \text{h.c.}) +\sum_{i=1}^{L} \mu_{i}\hat{n}_i,
\end{equation}
where $\hat{n}_i=b_{i,1}^{\dagger} b_{i,1}+b_{i,2}^{\dagger} b_{i,2}$ is the number operator. This model has strictly ICPF $\succ$ CA, which exhibits qHSF (cHSF) under the ICPF (CA) method. In the SM \cite{SM}, we explicitly show that at $L=3$, the CA Krylov subspaces of the QB model spanned by product state basis can be further fragmented by ICPF into smaller Krylov subspaces spanned by an entangled state basis, thereby exhibiting qHSF. The ICPF Krylov subspace bases depend on the coefficients $\{J_{i}^{\nu}\}$, in agreement with \cref{proposition: when ICPF > CA}.

\emph{(iv) The Fibonacci matrix (FM) model}, which has a Hamiltonian defined on a 1D lattice with 2 basis states $(|0\rangle_i,|1\rangle_i)$ per site $i$, with parameters $\{J_i\}$:
\begin{equation}
\label{eqn: Hamiltonian of FM}
H_{\text{FM}} 
= \sum_{i=1}^{L}
J_i \hat{F}_i\ ,\ \ 
\hat{F}_i= |0\rangle\langle 0|_i+|0\rangle\langle 1|_i+|1\rangle\langle 0|_i\ .
\end{equation}
The $n$-th Fibonacci number is efficiently generated as the $|0\rangle\langle 0|_i$ entry of $\hat{F}_i^{n-1}$. This model has $\mathcal{A}=\mathcal{C}=\mathcal{Z}$ generated by all $\hat{F}_i$. Thus, $\bold{Basis}(\mathcal{Z})=\{\hat{F}_i\}$ has irrational eigenvalues $\frac{1\pm\sqrt{5}}{2}$, and \cref{theorem: main theorem} does not apply. The CA method gives a qHSF into $2^L$ one-dimensional Krylov subspaces, namely the eigenstates. In contrast, the ICPF method yields $2^{L-1}$ two-dimensional Krylov subspaces as one can verify. So the model has ICPF $\prec$ CA.

\emph{Discussion.} We have proved \cref{theorem: main theorem} for Hamiltonians consisting of rational operators in \cref{eqn: Hamiltonian in J}, which ensures ICPF $\succeq$ CA under the condition $\bold{Basis}(\mathcal{Z})$ have all eigenvalues rational. Interestingly, this condition is satisfied by most known models, except for relatively trivial models such as the FM model constructed here. This bridges the ICPF and CA methods for identifying HSF, which originate from vastly distinct ideas. Additionally, \cref{proposition: when ICPF > CA} reveals the parameter dependence of the Krylov bases as a reason for ICPF and CA to give different HSFs. This allows one to apply the numerically more efficient ICPF method to identify the CA Krylov subspaces. A potential future direction is to extend the framework to certain classes of irrational operators, e.g. operators with algebraic number matrix representations, which would include more models such as spin models of spin one or higher. Moreover, it would be interesting and useful to further explore deeper connections between the two methods in quantum entanglement, dynamical localization and preventing thermalization, which may eventually lead to a unified definition of HSF.

\begin{acknowledgments}
\emph{Acknowledgments}. B.~C. is particularly grateful to Abhinav Prem and Nicolas Regnault for helping shape the initial ideas, and thanks Yumin Hu, Tian-Hua Yang, and Zihan Zhou for helpful discussions. This work is supported by the National Science Foundation under award DMR-2141966. 
\end{acknowledgments}

\bibliography{mybib.bib}

\clearpage
\onecolumngrid

\setcounter{equation}{0}
\setcounter{figure}{0}
\setcounter{table}{0}
\setcounter{page}{1}
\makeatletter
\renewcommand{\theequation}{S\arabic{equation}}
\renewcommand{\thefigure}{S\arabic{figure}}
\renewcommand{\thetable}{S\arabic{table}}

\begin{center}
{\bf \large Supplemental Material}
\end{center}

\section{Basic properties of von Neumann algebras}

We introduce some definitions and state several theorems without proof; their proofs can be found in existing references, such as \cite{Harlow:2016vwg}.

\begin{definition}
    Let $\mathcal{H}$ be a finite-dimensional Hilbert space, and let $\mathcal{L}(\mathcal{H})$ be the set of linear operators that act on $\mathcal{H}$. A von Neumann algebra on $\mathcal{H}$ is a set of operators $\mathcal{A} \subseteq \mathcal{L}(\mathcal{H})$ that satisfies:
    \begin{itemize}
        \item $\forall \lambda \in \mathbb{C}$, $\lambda \mathds{1} \in \mathcal{A}$.
        \item $\forall x \in \mathcal{A}$, $x^{\dagger} \in \mathcal{A}$.
        \item $\forall x, y \in \mathcal{A}$, $x+y \in \mathcal{A}$, $xy \in \mathcal{A}$.
    \end{itemize}
\end{definition}

\begin{definition}
  Given a von Neumann algebra ${\mathcal{A}}$, its commutant algebra $\mathcal{C}$ is defined as the set of operators on $\mathcal{H}$ that commutes with every element of ${\mathcal{A}}$. The center of ${\mathcal{A}}$ is defined as ${\mathcal Z} = \mathcal{A} \cap \mathcal{C}$. It is straightforward to verify that both $\mathcal{C}$ and ${\mathcal Z}$ are themselves von Neumann algebras.
\end{definition}

\begin{theorem}[The bi-commutant theorem]
  Let $\mathcal{A}$ be a von Neumann algebra, and $\mathcal{C}$ be its commutant algebra. The bi-commutant theorem states that the commutant algebra of $\mathcal{C}$ is equal to $\mathcal{A}$ itself.
\end{theorem}

\begin{theorem}
\label{theoremSM: block-decomposition A C}
    Let $\mathcal{A}$ be a von Neumann algebra acting on a Hilbert space $\mathcal{H}$, and let $\mathcal{C}$ denote its commutant. Then $\mathcal{H}$ admits a block decomposition of the form
    \begin{equation}
        \mathcal{H} = \bigoplus_\lambda \left( \mathcal{H}_{\lambda}^{(\mathcal{A})} \otimes \mathcal{H}_{\lambda}^{(\mathcal{C})} \right)
    \end{equation}
    where $\mathcal{H}_{\lambda}^{(\mathcal{A})}$ and $\mathcal{H}_{\lambda}^{(\mathcal{C})}$ are irreducible representations of ${\mathcal{A}}$ and $\mathcal{C}$ respectively. With respect to this decomposition, both $\mathcal{A}$ and $\mathcal{C}$ take block-diagonal forms,
    \begin{equation}
    \label{eqnSM: decom H into HA and HC}
        \mathcal{A} = \bigoplus_\lambda \left( \mathcal{L}(\mathcal{H}_{\mathcal{A}}) \otimes \mathds{1} \right)
        \;, \;
        \mathcal{C} = \bigoplus_\lambda \left( \mathds{1} \otimes \mathcal{L}(\mathcal{H}_{\mathcal{C}}) \right).
    \end{equation}
    The blocks labeled by $\lambda$ are refered to as summands.
\end{theorem}

\section{Formal proof of ICPF $\succeq$ CA}
\label{sectionSM: formal proof of ICPF <= CA}

\subsection{ICPF $\succeq$ CA}

In the main text, we briefly sketched the proof of ICPF $\succeq$ CA. Here, we present the complete proof.
We first introduce some terminologies.
Operators whose matrix representations contain only rational number entries in the Hamiltonian basis are referred to as \emph{rational operators}. States (not necessarily normalized) that admit integer coefficient decompositions in the Hamiltonian basis are referred to as \emph{rational states}; when normalized, the ratios between their components are rational.

\begin{theorem}
	\label{theoremSM: Z is rational under u}
	In the Hamiltonian basis $\{\ket{u_i}\}$, where all $\hat{h}_{\eta}$ are rational operators, every elements of $\bold{Basis}(\mathcal{Z})=\{ \hat{z}_{k} \}$ can be expressed as rational operators.
\end{theorem}
\begin{proof}
	We start by constructing a linear-independent basis $\bold{Basis}(\mathcal{A})$ through an iterative procedure:
	\begin{enumerate}
		\item Initialize a set $S$ with the identity $S = \{ \mathds{1} \}$. 
		\item Expand $S$ by multiplying each of its elements with every $\hat{h}_\eta$, adding the products to $S$:
		      \begin{equation}
			      S \leftarrow S \, \cup \left\{ \hat{s} \cdot \hat{h}_{\eta} \mid \hat{s}\in S, \; \forall \hat{h}_\eta \right\}
		      \end{equation}
		\item Remove any linearly dependent element from $S$.
		\item Repeat step 2 and step 3 until the set $S$ no longer changes---that is, no new linearly independent elements are added. The resulting set $S$ is taken as $\bold{Basis}(\mathcal{A})$.
	\end{enumerate}
	Since this procedure involves only rational arithmetic, all elements in $\bold{Basis}(\mathcal{A})=\{ \hat{a}_{k} \}$ are rational operators. To construct $\bold{Basis}(\mathcal{Z})$, we solve for the coefficients $w_{k}$ satisfying the commuting condition:
	\begin{equation}\label{seq:basisZ}
		\forall \hat{h}_{\eta}, \quad \left[ \sum\nolimits_{j} w_{j} \hat{a}_{j}, \hat{h}_{\eta} \right] = 0
	\end{equation}
	Each linearly independent solution $w_j=w_{j}^{(k)}$ (the $k$-th solution) corresponds to an element $\hat{z}_{k}=\sum_jw_{j}^{(k)}\hat{a}_j$ in $\bold{Basis}(\mathcal{Z})=\{\hat{z}_{k}\}$. We denote the number of linearly independent solutions as $Z=|\bold{Basis}(\mathcal{Z})|$. Since \cref{seq:basisZ} are a set of equations of $\{w_{j}\}$ with rational coefficients, they admit rational solutions, and thus the resulting basis operators $\{\hat{z}_{k}\}$ can all be taken to be rational operators. Note that $\bold{Basis}(\mathcal{Z})$ is not uniquely determined; however, we are only asserting the existence of \textit{at least one} such rational basis.
\end{proof}

Since the operators $\hat{z}_{k}$ commute with one another, they can be simultaneously diagonalized. The full Hilbert space then decomposes into the simultaneous eigenspaces $V_\lambda$ of $\{z_{k}\}$:
\begin{equation}
\label{eqnSM: decom of H into V lambda}
    {\cal H} = \bigoplus_{\lambda} V_{\lambda} \;, \quad  \forall \ket{v} \in V_{\lambda}, \quad \hat{z}_{k}\ket{v} = c_\lambda(\hat{z}_k) \ket{v}
,\end{equation}
where $c_\lambda(\hat{z}_k)$ are the eigenvalue of $\hat{z}_k$ on eigenspace $V_\lambda$. The label $\lambda$ indexes all distinct joint eigenspaces, such that for any $\lambda\neq \lambda^\prime$, there exists at least one $\hat{z}_k$ such that $c_\lambda(\hat{z}_k)\neq c_{\lambda^\prime}(\hat{z}_k)$. According to \cref{theoremSM: block-decomposition A C}, both $\mathcal{A}$ and $\mathcal{C}$ admit block-diagonal representations, where the blocks (referred to as summands) correspond to the subspaces  $V_\lambda=\mathcal{H}^{(\mathcal{A})}_\lambda \otimes \mathcal{H}^{(\mathcal{C})}_\lambda$. The projections onto these subspaces are denoted by $\hat{P}_\lambda$. These projections can be constructed as a product of linear combinations of $\hat{z}_k$:
\begin{equation}
\label{eqnSM: explicit expression of P lambda k}
    \projection_{\lambda}
    =
    \prod_{k}
    \left(
    \prod_{\lambda^{\prime}\in \Lambda(\lambda,k)}
    \frac
    {\hat{z}_k - c_{\lambda^{\prime}}(\hat{z}_k)}
    {c_{\lambda}(\hat{z}_k) - c_{\lambda^{\prime}}(\hat{z}_k)}
    \right)
,\end{equation}
where we have defined $\Lambda(\lambda,k)=\{\lambda'\ |\  c_{\lambda^{\prime}}(\hat{z}_k) \neq c_{\lambda}(\hat{z}_k)\}$ to be the set of summands $\lambda'$ with $\hat{z}_k$ eigenvalue different from that of $\lambda$.

\begin{theorem}
     The projections $\hat{\mathcal{P}}_{\lambda}$ are rational operators if and only if all the eigenvalues of each element in $\bold{Basis}(\mathcal{Z})$ are rational.
\end{theorem}
\begin{proof}
    The eigenvalues of $\hat{z}_k\in\bold{Basis}(\mathcal{Z})$ are given by $c_{\lambda}(\hat{z}_k)$. If all these coefficients are rational, and if, as established in \cref{theoremSM: Z is rational under u}, the operators $\hat{z}_k$ are rational operators, then the projections $\hat{\mathcal{P}}_{\lambda}$ are rational operators.

    Conversely, suppose that for a particular $\lambda$ and $k$, the eigenvalue $c_{\lambda}(\hat{z}_{k})$ is irrational. Then, from \eqref{eqnSM: decom of H into V lambda}, we have:
    \begin{equation}
        \hat{z}_{k} \projection_{\lambda}
        = c_{\lambda}(\hat{z}_{k}) \projection_{\lambda}
    \end{equation}
    If the projection $\hat{\mathcal{P}}_{\lambda}$ is rational, the left-hand side of the equation would only contain rational numbers. However, the right-hand side involves the irrational number $c_{\lambda}(\hat{z}_{k})$, leading to a contradiction. Therefore, $\hat{\mathcal{P}}_{\lambda}$ must contain irrational entries under the Hamiltonian basis.
\end{proof}

Under the Hamiltonian basis, both the Hamiltonian $H(J_\eta)$ and the projections $\hat{\mathcal{P}}_{\lambda}$ have rational matrix representations. As a result, the characteristic polynomial $\Phi_{\lambda}(\mu) = \det_\lambda(\mu \mathds{1} - \hat{\mathcal{P}}_{\lambda} H(J_{\eta}) \hat{\mathcal{P}}_{\lambda})$ has rational coefficients. In other words, the Hamiltonian, when projected onto each summand (labeled by $\lambda$), has a rational characteristic polynomial. Consequently, the ICPF approach can capture the invariant subspaces with different $\lambda$.

In a summand labeled by $\lambda$, the (monic) characteristic polynomial $\Phi_{\lambda}(\mu)$ can be expressed as a product of sub-polynomial raised to the power $d_\lambda$: $\Phi_{\lambda}(\mu) = \left[ \phi_{\lambda}^{(D_\lambda)}(\mu) \right]^{d_{\lambda}}$, where the sub-polynomial $\phi_{\lambda}^{(D_\lambda)}(\mu)$ corresponds to a degenerate subspace of dimension $D_\lambda$. What remains to be proven is whether $\phi_{\lambda}^{(D_\lambda)}(\mu)$ has rational coefficients.

\begin{theorem}
	Let $f(x)$ be a monic polynomial, meaning its leading coefficient is 1. If $[f(x)]^{n}$ has rational coefficients for some natural number $n$, then $f(x)$ must also have rational coefficients.
\end{theorem}
\begin{proof}
	Let $f(x) = x^{d} + a_{d-1}x^{d-1} + \cdots a_{0}$. Expanding $[f(x)]^{n}$, we have
	\begin{align}
		f^{n}(x) = & x^{nd} + na_{n-1}x^{nd -1} +
        \left[na_{n-2} + \frac{n(n-1)}{2}a_{n-1}^{2}\right]x^{nd-2} + \cdots
	\end{align}
	Let's examine the coefficients from the highest degree terms:
	\begin{enumerate}
		\item The coefficient of $x^{nd -1}$ is $n a_{n-1}$, which is rational; therefore $a_{n-1}$ is rational.
		\item Assume coefficients $a_{n-i + 1}, \cdots a_{n-1}$ are rational for some $i\ge 2$.
		      The coefficient of $x^{nd-i}$ term is $na_{n-i} + c_{n-i}$, where $c_{n-i}$ is a rational polynomial involving higher-order coefficients $a_{n-i+1}, \cdots, a_{n-1}$. The expression $c_{n-1}$ is rational by the inductive hypothesis, therefore $a_{n-i}$ is rational.
	\end{enumerate}
	By induction, all coefficients $a_{i}$ for $i\in[0, d-1]$ are rational. Therefore, $f(x)$ has rational coefficients.

	There is a simpler proof using Galois theory: Let $\mathbb{K} = \mathbb{Q}(a_{n-1}, \cdots a_{0})$ be the field extensions generated by the coefficients $a_{i}$. Let $\sigma \in \text{Gal}(\mathbb{K}/\mathbb{Q})$ be an automorphism of the Galois group of the field extension $\mathbb{K}/\mathbb{Q}$. Since $f^{n}$ has rational coefficients, we have $\sigma(f^{n}) = \sigma(f)^{n} = f^{n}$. Because $f$ is monic, this implies $\sigma(f) = f$. As a result, $\sigma$ fixes all coefficients $a_{i}$, and so every automorphism in the Galois group acts trivially. Therefore, $\text{Gal}(\mathbb{K}/\mathbb{Q})$ is trivial. By the fundamental theorem of Galois theory \cite{stewart2022galois}, $\mathbb{K} = \mathbb{Q}$, and all coefficients $a_{i}$ are rational.
\end{proof}
\noindent
This completes the proof that $\phi_{\lambda}^{(D_\lambda)}(\mu)$ has rational coefficients, which corresponds to the characteristic polynomial of each Krylov subspace identified by the CA method. We have now proved the main theorem in the main text, which we briefly restate:
\begin{theorem}[\textbf{The main theorem of this paper}]
	If all the eigenvalues of each element in $\bold{Basis}(\mathcal{Z})$ are rational, then ICPF $\succeq$ CA.
\end{theorem}

\noindent

\subsection{ICPF $\succ$ CA: coefficient-dependent basis}

In models such as the quantum breakdown model, the HSFs from the two methods are different, and the inequality in our main theorem is not saturated. The ICPF method further decomposes the CA Krylov subspaces into smaller components, namely the ICPF subspaces. In such cases, the finer ICPF subspaces must depend on the Hamiltonian coefficients $\{J_\eta\}$. This is stated in Proposition 1 of the main text, which can be proved as follows.

Suppose that a CA Krylov subspace decomposes into a direct sum of ICPF Krylov subspaces: $\mathcal{K}_{\text{CA}} = \bigoplus_i \mathcal{K}_{\text{ICPF}}^{(i)}$, and assume that the basis of ${\cal K}_{\text{ICPF}}^{(i)}$ is independent of $\{J_\eta\}$. This implies that each subspace is invariant under the action of the Hamiltonian for any choice of coefficients:
\begin{equation}
H(\{J_\eta\})\ket{v}\in {\cal K}_{\text{ICPF}}^{(i)}, \quad \forall \ket{v} \in {\cal K}_{\text{ICPF}}^{(i)}.
\end{equation}
Since ${\cal K}_{\text{ICPF}}^{(i)}$ are assumed to be independent of $\{J_\eta\}$, we may set $H(\{J_\eta\}) = \hat{h}_{\eta}$, yielding:
\begin{equation}
\hat{h}_{\eta}\ket{v}\in {\cal K}_{\text{ICPF}}^{(i)}, \quad \forall \eta, \;\ket{v} \in {\cal K}_{\text{ICPF}}^{(i)}.
\end{equation}
This means that the subspace ${\cal K}_{\text{ICPF}}^{(i)}$ is invariant under the action of all local terms $\hat{h}_\eta$, implying that $\mathcal{L}(\mathcal{H}_{\mathcal{A}})$ in \eqref{eqnSM: decom H into HA and HC} is reducible. This contradicts the assumption that $\mathcal{H}_{\mathcal{A}}$ supports an irreducible representation.

\section{$t$-$J_{z}$ model}

\subsection{Model description}
The $t$-$J_{z}$ model~\cite{PhysRevB.55.6491, PhysRevLett.85.4755} is a 1-dimensional lattice model that describes spinful fermions, with their creation operators at the $i$th site denoted as $c^{\dagger}_{i,\uparrow}$, $c^{\dagger}_{i, \downarrow}$ respectively.  Given a system of size $L$, the Hamiltonian is
\begin{equation}
H_{t-J_z} = \sum_{i=1}^{L-1}\left[ -t_{i, i+1} \sum_{\sigma \in \{\uparrow, \downarrow\} }\left(\tilde{c}_{i, \sigma}\tilde{c}^{\dagger}_{i + 1, \sigma} + \textrm{h.c.}\right) + J_{i, i+1}S_{i}^{z}S_{i + 1}^{z} \right]+ \sum_{i=1}^{L}\left[ h_i S^{z}_{i}  + g_i(S^{z}_{i})^2\right],
\end{equation}
where the modified fermionic operators are given by $\tilde{c}_{i, \sigma} \equiv c_{i}^{\sigma}(1 - c_{i, -\sigma}^{\dagger}c_{i, -\sigma})$, and 
$S^{z}_i \equiv \tilde{c}_{i, \uparrow}\tilde{c}^{\dagger}_{i, \uparrow} - \tilde{c}_{i, \downarrow}\tilde{c}^{\dagger}_{i, \downarrow} $. The Hilbert space of the model excludes double occupancy of fermions at the same site: each site can either be empty, occupied by a spin-up fermion, or occupied by a spin-down fermion. Therefore, the total dimension of the Hilbert space is $3^L$. The single-site operators 
$S^{z}_{i}$ and $(S^{z}_{i})^2$ break symmetries that are not of interest. 

The $t$-$J_z$ model has two (local) integrals of motion, which correspond to the conserved particle numbers of the two fermion species $N_{\uparrow}$, $N_{\downarrow}$:
\begin{equation}
\label{eqnSM: conserved quantities of tJ model}
N^{\uparrow} = \sum_{i =1 }^{L}\tilde{c}_{i, \uparrow}\tilde{c}^{\dagger}_{i , \uparrow}, \quad N^{\downarrow} = \sum_{i =1 }^{L}\tilde{c}_{i, \downarrow}\tilde{c}^{\dagger}_{i , \downarrow}.
\end{equation}
Summing over all possible values of $N_{\uparrow}$ and $N_{\downarrow}$ recovers the total dimension of the Hilbert space $3^L$:
\begin{equation}
\sum_{N_{\uparrow} =0}^{L}\sum_{N_{\downarrow} = 0}^{L - N_{\uparrow}} \frac{L!}{N_{\uparrow}!N_{\downarrow}!(L - N_{\uparrow}! -N_{\downarrow}!)} = (1 + 1 + 1)^L = 3^L.
\end{equation}

\subsection{Hilbert space fragmentation}

The nontrivial actions of the Hamiltonian involve shifting spin-up and spin-down states to the left and the right:
\begin{equation}
\ket{\cdots \, 0 \uparrow\, \cdots} \leftrightarrow \ket{\cdots \,  \uparrow 0 \, \cdots}, \quad
\ket{\cdots \, 0 \downarrow\, \cdots} \leftrightarrow \ket{\cdots \,  \downarrow 0 \, \cdots}.
\end{equation}
Importantly, a spin-up state cannot pass through a spin-down state, because double occupancy is excluded. Therefore, under OBC, different sectors of the Krylov subspaces are uniquely determined by the ordering of $\ket{\uparrow}$ and $\ket{\downarrow}$ in the spin chain. For example, $\ket{\uparrow\,\downarrow 0 \uparrow 0}$ and $\ket{\uparrow\,0\downarrow  0 \uparrow }$ belong to the same Krylov subspace, while $\ket{\uparrow 0 \downarrow 0 \uparrow }$ and $\ket{\uparrow 0\uparrow 0 \downarrow }$, although belonging to the same symmetry sector, lie in different Krylov subspaces. In each sector, there are $d_{N_{\uparrow}, N_{\downarrow}} = \frac{(N_{\uparrow} + N_{\downarrow})!}{N_{\uparrow}!N_{\downarrow}}$ ways to order $\uparrow$ and $\downarrow$, so the $(N_{\uparrow}, N_{\downarrow})$ sector decomposes into $d_{N_{\uparrow},N_{\downarrow}}$ distinct Krylov subspaces, each of which has dimension $D_{N_{\uparrow}, N_{\downarrow}} = \frac{L!}{(N_{\uparrow} + N_{\downarrow})!(L-N_{\uparrow}- N_{\downarrow})!}$.

The Krylov subspaces described above can be captured using the ICPF approach by substituting integer coefficients into the Hamiltonian. As an example, let's consider the model with $L=4$. Choose a random set of integer coefficients
$$
\begin{cases}
    (t_{1, 2}, t_{2, 3}, t_{3, 4}) &= (-3, -4, 2)
    \\
    (J_{1, 2}, J_{2, 3}, J_{3, 4}) &= (-1, -4, 5)
    \\
    (h_1, h_2, h_3, h_4) &= (2, 5, -5, -1)
    \\
    (g_1, g_2, g_3, g_4) &= (-3, -6, -3, 1)
\end{cases}
,$$
the ICPF approach characteristics the Hilbert space fragmentation by the factorization:
\begin{equation}
\begin{aligned}
    &\det(\mu \mathds{1} - H) 
    = \phi_1^{(4)}(\mu) \cdot \phi_2^{(4)}(\mu) \cdot \phi_3^{(4)}(\mu) \cdots
    \\
    =&
    (\mu^4 + 2 \mu^3 - 256 \mu^2 - 362 \mu + 
   12327) (\mu^4 + 36 \mu^3 + 324 \mu^2 + 
   48 \mu - 1665) (\mu^4 + 64 \mu^3 + 
   1460 \mu^2 + 13684 \mu + 41607) \cdots
,\end{aligned}
\end{equation}
where we only explicitly show the factors corresponding to the symmetry sector with $N_\uparrow = 1$ and $N_\downarrow=2$, which has dimension 12, while the factors corresponding to other symmetry sectors are abbreviated as $\cdots$. The superscript $D_p$ in the factor $\phi_{p}^{(D_p)}(\mu)$ denotes the highest order of $\mu$, while the subscript $i$ serves as a label distinguishing different factors. According to the ICPF approach, the symmetry sector decomposes into three distinct Krylov subspaces, each of dimension four. The Krylov subspaces are spanned by specific sets of basis:
\begin{equation}
\label{eqnSM: ICPF of the t-J model basis}
    \begin{cases}
        \phi_1^{(4)}(\mu):& 
        \ket{\,\uparrow \,\downarrow\, \downarrow 0}, 
        \ket{\,\uparrow \,\downarrow 0\downarrow }, 
        \ket{\,\uparrow 0 \downarrow\, \downarrow},
        \ket{0\uparrow \,\downarrow\, \downarrow}
        \\
        \phi_2^{(4)}(\mu):& 
        \ket{\,\downarrow \,\uparrow\, \downarrow 0},
        \ket{\,\downarrow \,\uparrow 0\downarrow }, 
        \ket{\,\downarrow 0 \uparrow\, \downarrow},
        \ket{0\downarrow \,\uparrow\, \downarrow}
        \\
        \phi_3^{(4)}(\mu):& 
        \ket{\,\downarrow \,\downarrow\, \uparrow 0},
        \ket{\,\downarrow \,\downarrow 0\uparrow }, 
        \ket{\,\downarrow 0 \downarrow\, \uparrow},
        \ket{0\downarrow \,\downarrow\, \uparrow}
    \end{cases}
.\end{equation}
We denote the root state---any one of the basis vectors---by $\ket{\Psi_p}$, which satisfies $\phi_{p}^{(4)}(H) \ket{\Psi_p} = 0$:
\begin{equation}
\label{eqnSM: ICPF of the t-J model characteristic poly}
    \begin{cases}
        \phi_1^{(4)}(\mu):& \ket{\Psi_1}=\ket{\,\uparrow \,\downarrow\, \downarrow 0}
        \longrightarrow
        (H^4 + 2 H^3 - 256 H^2 - 362 H + 12327) \ket{\Psi_1} = 0
        \\
        \phi_2^{(4)}(\mu):& \ket{\Psi_2}=\ket{\,\downarrow \,\uparrow\, \downarrow 0}
        \longrightarrow
        (H^4 + 36 H^3 + 324 H^2 + 48 H - 1665)\ket{\Psi_2} = 0
        \\
        \phi_3^{(4)}(\mu):& \ket{\Psi_3}=\ket{\,\downarrow \,\downarrow\, \uparrow 0}
        \longrightarrow
        (H^4 + 64 H^3 + 1460 H^2 + 13684 H + 41607)\ket{\Psi_3} = 0
    \end{cases}
.\end{equation}
These root states, under the action of the Hamiltonian, generate their corresponding 4 dimensional Krylov subspaces.

On the other hand, we can study HSF using the CA approach. The commutants of the $t$-$J_z$ model are
\begin{equation}
    N^{\sigma_1 \sigma_2 \cdots \sigma_k}
    =
    \sum_{j_1 <j_2 <\cdots<j_k} N_{j_1}^{\sigma_1} N_{j_2}^{\sigma_2} \cdots N_{j_k}^{\sigma_k}
,\end{equation}
where $N_{j}^{\sigma}=\tilde{c}_{j, \sigma}\tilde{c}^{\dagger}_{j , \sigma}$ and $1 \le k \le L$. In particular, $k=1$ corresponds to the conserved charges introduced in \cref{eqnSM: conserved quantities of tJ model}. The commutants take distinct (one-dimensional) representations within each Krylov subspace. In the symmetry sector $(N_\uparrow, N_\downarrow)=(1, 2)$, each Krylov subspace---represented by one of its basis states---carries the corresponding representations of the commutants:
\begin{equation}
\label{eqnSM: commutant table of the t-J model}
\begin{array}{|c|c|c|c|c|c|c|c|c|c|c|c|c|c|c|}
        \hline
         & N^{\uparrow} & N^{\downarrow} & N^{\uparrow\uparrow} & N^{\uparrow\downarrow} & N^{\downarrow\uparrow} & N^{\downarrow\downarrow} & N^{\uparrow\uparrow\uparrow} & N^{\uparrow\uparrow\downarrow} & N^{\uparrow\downarrow\uparrow} & N^{\uparrow\downarrow\downarrow} & N^{\downarrow\uparrow\uparrow} & N^{\downarrow\uparrow\downarrow} & N^{\downarrow\downarrow\uparrow} & N^{\downarrow\downarrow\downarrow} \\ \hline
        \ket{\,\uparrow \,\downarrow\, \downarrow 0} & 1 & 2 & 0 & 2 & 0 & 1 & 0 & 0 & 0 & 1 & 0 & 0 & 0 & 0 \\ \hline
        \ket{\,\downarrow \,\uparrow\, \downarrow 0} & 1 & 2 & 0 & 1 & 1 & 1 & 0 & 0 & 0 & 0 & 0 & 1 & 0 & 0 \\ \hline
        \ket{\,\downarrow \,\downarrow\, \uparrow 0} & 1 & 2 & 0 & 0 & 2 & 1 & 0 & 0 & 0 & 0 & 0 & 0 & 1 & 0 \\ \hline
\end{array}
\end{equation}

Comparing \cref{eqnSM: ICPF of the t-J model characteristic poly} and \cref{eqnSM: commutant table of the t-J model}, we see that in the $(N_{\uparrow}, N_{\downarrow}) = (1, 2)$ sector, the ICPF approach and the CA approach yield equivalent definitions of the Krylov subspaces $L = 4$ $t$-$J_z$ model. As a result, both approaches lead to the same HSF structure. This equivalence extends to other symmetry sectors as well, and the entire Hilbert space is fragmented identically under both approaches. More generally, we found that the ICPF approach coincides with the CA approach for all system sizes up to $L \leq 8$.

\section{Pair-Flip model}

\subsection{Model description}
The pair-flip model~\cite{caha2018pairflip} describes a 1-dimensional lattice model with size $L$, each site corresponds to an $m$-level system. The evolution of the model exhibits HSF under the Hamiltonian
\begin{equation}
H_{\text{PF}} = \sum_{j=1}^{L-1} \sum_{\alpha, \beta=1}^{m} 
    g_{j}^{\alpha, \beta}
    \underbrace{
    \left( \rket{\alpha \alpha} \lket{\beta \beta} \right)_{j, j+1}
    }_{F_{j,j+1}^{\alpha,\beta}}
    +
    \sum_{j=1}^{L} \sum_{\alpha=1}^{m} 
    \mu_{j}^{\alpha}
    \underbrace{\left( \rket{\alpha} \lket{\alpha} \right)_{j}}_{N_{j}^{\alpha}}
.\end{equation}
Here, $g_{i}^{\alpha, \beta}={g_{i}^{\beta, \alpha}}^*$ to ensure hermiticity.

The system has $U(1)$ conserved quantities:
\begin{equation} 
N^{\alpha} = \sum_{j}(-1)^{j}N^{\alpha}_{j}.
\end{equation}
However, these quantities are not independent; they satisfy the relation $\sum_{\alpha}N^{\alpha} = \mathds{1}$. Therefore, the total global symmetry is $U(1)^{m-1}$.

For models with $m=3$, we introduce a color-based notation:
\[
\ket{1} = \tket{\node[circle, fill=red, inner sep=1.5pt] {};}, \quad
\ket{2} = \tket{\node[circle, fill=mydarkgreen, inner sep=1.5pt] {};}, \quad
\ket{3} = \tket{\node[circle, fill=blue, inner sep=1.5pt] {};}. 
\]
To represent product states, we construct the graphs in the following way: we start with unconnected colored dots. Then, neighboring dots of the same color are connected to form dimers, which are subsequently excluded from further consideration. This process is repeated with the remaining unconnected dots until no additional dimers can be made. The final graph is composed of isolated dots and dimers. For example:
\[
\ket{323321} =
\ket{\vphantom{\int}\;\tikz[baseline={([yshift=-1.4ex]current bounding box.center)}]{
    \node[circle, fill=blue, inner sep=1.5pt] at (0,0) {};
    \node[circle, fill=mydarkgreen, inner sep=1.5pt] (A) at (0.5,0) {};
    \node[circle, fill=blue, inner sep=1.5pt] (B) at (1,0,0) {};
    \node[circle, fill=blue, inner sep=1.5pt] (C) at (1.5,0) {};
    \node[circle, fill=mydarkgreen, inner sep=1.5pt] (D) at (2,0,0) {};
    \node[circle, fill=red, inner sep=1.5pt] at (2.5,0) {};
    \draw[thick, mydarkgreen] (A) to[out=30, in=150] (D); 
    \draw[thick, blue] (B) to (C); 
}\;}.
\]

\subsection{Hilbert space fragmentation}

By applying $F^{\alpha, \beta}_{i, i+1}$ to pure product states, any configuration of dimers can be transformed into any set of non-crossing dimers of arbitrary colors. For example:
\[
\ket{\vphantom{\int}\;\tikz[baseline={([yshift=-1.4ex]current bounding box.center)}]{
    \node[circle, fill=red, inner sep=1.5pt] (A) at (0,0) {};
    \node[circle, fill=blue, inner sep=1.5pt] (B) at (0.5,0,0) {};
    \node[circle, fill=blue, inner sep=1.5pt] (C) at (1,0) {};
    \node[circle, fill=red, inner sep=1.5pt] (D) at (1.5,0,0) {};
    \draw[thick, red] (A) to[out=30, in=150] (D); 
    \draw[thick, blue] (B) to (C); 
}\;} 
\leftrightarrow
\ket{\vphantom{\int}\;\tikz[baseline={([yshift=-1.4ex]current bounding box.center)}]{
    \node[circle, fill=red, inner sep=1.5pt] (A) at (0,0) {};
    \node[circle, fill=red, inner sep=1.5pt] (B) at (0.5,0,0) {};
    \node[circle, fill=red, inner sep=1.5pt] (C) at (1,0) {};
    \node[circle, fill=red, inner sep=1.5pt] (D) at (1.5,0,0) {};
    \draw[thick, red] (A) to[out=30, in=150] (D); 
    \draw[thick, red] (B) to (C); 
}\;} =  
\ket{\vphantom{\int}\;\tikz[baseline={([yshift=-0.5ex]current bounding box.center)}]{
    \node[circle, fill=red, inner sep=1.5pt] (A) at (0,0) {};
    \node[circle, fill=red, inner sep=1.5pt] (B) at (0.5,0,0) {};
    \node[circle, fill=red, inner sep=1.5pt] (C) at (1,0) {};
    \node[circle, fill=red, inner sep=1.5pt] (D) at (1.5,0,0) {};
    \draw[thick, red] (A) to (B); 
    \draw[thick, red] (C) to (D); 
}\;} \leftrightarrow 
\ket{\vphantom{\int}\;\tikz[baseline={([yshift=-0.5ex]current bounding box.center)}]{
    \node[circle, fill=mydarkgreen, inner sep=1.5pt] (A) at (0,0) {};
    \node[circle, fill=mydarkgreen, inner sep=1.5pt] (B) at (0.5,0,0) {};
    \node[circle, fill=mydarkgreen, inner sep=1.5pt] (C) at (1,0) {};
    \node[circle, fill=mydarkgreen, inner sep=1.5pt] (D) at (1.5,0,0) {};
    \draw[thick, mydarkgreen] (A) to (B); 
    \draw[thick, mydarkgreen] (C) to (D); 
}\;}.
\]
In addition, any unpaired dot can move past dimers under $F^{\alpha, \beta}_{i, i+1}$. For example
\[
\tket{
    \draw[thick, blue] (0.5,0) -- (1, 0);
    \node[circle, fill=red, inner sep=1.5pt] at (0,0) {};
    \node[circle, fill=blue, inner sep=1.5pt] at (0.5,0) {};
    \node[circle, fill=blue, inner sep=1.5pt] at (1,0) {};
} \leftrightarrow
\tket{
    \draw[thick, red] (0.5,0) -- (1, 0);
    \node[circle, fill=red, inner sep=1.5pt] at (0,0) {};
    \node[circle, fill=red, inner sep=1.5pt] at (0.5,0) {};
    \node[circle, fill=red, inner sep=1.5pt] at (1,0) {};
} =
\tket{
    \draw[thick, red] (0,0) -- (0.5, 0);
    \node[circle, fill=red, inner sep=1.5pt] at (0,0) {};
    \node[circle, fill=red, inner sep=1.5pt] at (0.5,0) {};
    \node[circle, fill=red, inner sep=1.5pt] at (1,0) {};
} \leftrightarrow
\tket{
    \draw[thick, blue] (0,0) -- (0.5, 0);
    \node[circle, fill=blue, inner sep=1.5pt] at (0,0) {};
    \node[circle, fill=blue, inner sep=1.5pt] at (0.5,0) {};
    \node[circle, fill=red, inner sep=1.5pt] at (1,0) {};
}.
\]
We observe that no matter how we apply $F^{\alpha, \beta}_{i, i+1}$ or $N^{\alpha}_i$ to the product states, the ordering of the singular dots remains unchanged. Therefore, different Krylov subspaces can be distinguished by the orderings of the singular dots.

Using the ICPF approach, these Krylov subspaces can be characterized by the factors obtained through polynomial factorization. Consider the Hamiltonian of a $L=6$ model with a random set of integer coefficients. The characteristic polynomial factorizes into multiple components: one factor of degree 87, 6 distinct factors of degrees 47, 24 distinct factors of degrees 11, and 96 distinct factors of degrees 1. We denote the corresponding polynomials by $\phi_{1}^{(87)}(\mu)$, $\phi_{2\sim7}^{(47)}(\mu)$, $\phi_{8\sim31}^{(11)}(\mu)$, $\phi_{32\sim127}^{(87)}(\mu)$, respectively. The basis vectors of the subspace corresponding to the factor $\phi_{1}^{(87)}(\mu)$ consist of fully connected dimers, i.e., without singular dots. The basis vectors associated with each $\phi_{p}^{(47)}(\mu)$ contain 2 connected dimers and 2 singular dots. Those corresponding to the factor $\phi_{p}^{(11)}(\mu)$ consist of 1 connected dimers and 4 singular dots. Finally, the basis vectors of the subspace corresponding to the factor $\phi_{p}^{(1)}(\mu)$ are entirely composed of singular dots.

The commutants of the model, which commutes with all $F^{\alpha, \beta}_{i, i+1}$ and $N^{\alpha}_i$, are:
\begin{equation}
\label{eqnSM: commutant generators of PF}
N^{\alpha_1 \cdots \alpha_{k}}
= \sum_{j_1 < \cdots  <j_k} (-1)^{\sum_{l=1}^{k}j_l}N_{j_1}^{\alpha_1}\cdots N_{j_k}^{\alpha_k}, \quad \alpha_{j} \neq \alpha_{j+1}
.\end{equation}
For $L=6$ model, the commutants decompose the Hilbert space into multiple irreps, each characterized by a distinct set of 1-dimensional representations of the commutant.
\begin{equation}
\label{eqnSM: commutant table of the PF model}
\begin{array}{|c|c|c|c|c|c|c|c|c|c|c|c|c|c|c|}
        \hline
         & N^{\bs} & N^{\gs} & N^{\rs} & N^{\bs\gs} & N^{\bs\rs} & N^{\gs\bs} & N^{\gs\rs} & N^{\rs\bs} & N^{\rs\gs} & N^{\bs\gs\rs} & D_\lambda & d_\lambda \\ \hline
        \tket{
    \draw[thick, blue] (0.0,0) -- (0.5, 0);
    \draw[thick, blue] (1.0,0) -- (1.5, 0);
    \draw[thick, blue] (2.0,0) -- (2.5, 0);
    \node[circle, fill=blue, inner sep=1.5pt] at (0,0) {};
    \node[circle, fill=blue, inner sep=1.5pt] at (0.5,0) {};
    \node[circle, fill=blue, inner sep=1.5pt] at (1,0) {};
    \node[circle, fill=blue, inner sep=1.5pt] at (1.5,0) {};
    \node[circle, fill=blue, inner sep=1.5pt] at (2,0) {};
    \node[circle, fill=blue, inner sep=1.5pt] at (2.5,0) {};
} 
        & 0 & 0 & 0 & 0 & 0 & 0 & 0 & 0 & 0 & 0 & 87 & 1 \\ \hline
    \tket{
    \draw[thick, blue] (0.0,0) -- (0.5, 0);
    \draw[thick, blue] (1.0,0) -- (1.5, 0);
    \node[circle, fill=blue, inner sep=1.5pt] at (0,0) {};
    \node[circle, fill=blue, inner sep=1.5pt] at (0.5,0) {};
    \node[circle, fill=blue, inner sep=1.5pt] at (1,0) {};
    \node[circle, fill=blue, inner sep=1.5pt] at (1.5,0) {};
    \node[circle, fill=blue, inner sep=1.5pt] at (2,0) {};
    \node[circle, fill=mydarkgreen, inner sep=1.5pt] at (2.5,0) {};
} 
        & -1 & 1 & 0 & -1 & 0 & 0 & 0 & 0 & 0 & 0 & 47 & 1 \\ \hline
            \tket{
    \draw[thick, blue] (0.0,0) -- (0.5, 0);
    \draw[thick, blue] (1.0,0) -- (1.5, 0);
    \node[circle, fill=blue, inner sep=1.5pt] at (0,0) {};
    \node[circle, fill=blue, inner sep=1.5pt] at (0.5,0) {};
    \node[circle, fill=blue, inner sep=1.5pt] at (1,0) {};
    \node[circle, fill=blue, inner sep=1.5pt] at (1.5,0) {};
    \node[circle, fill=blue, inner sep=1.5pt] at (2,0) {};
    \node[circle, fill=red, inner sep=1.5pt] at (2.5,0) {};
} 
        & -1 & 0 & 1 & 0 & -1 & 0 & 0 & 0 & 0 & 0 & 47 & 1 \\ \hline
    \tket{
    \draw[thick, blue] (0.0,0) -- (0.5, 0);
    \draw[thick, blue] (1.0,0) -- (1.5, 0);
    \node[circle, fill=blue, inner sep=1.5pt] at (0,0) {};
    \node[circle, fill=blue, inner sep=1.5pt] at (0.5,0) {};
    \node[circle, fill=blue, inner sep=1.5pt] at (1,0) {};
    \node[circle, fill=blue, inner sep=1.5pt] at (1.5,0) {};
    \node[circle, fill=mydarkgreen, inner sep=1.5pt] at (2,0) {};
    \node[circle, fill=blue, inner sep=1.5pt] at (2.5,0) {};
}
        & 1 & -1 & 0 & 0 & 0 & -1 & 0 & 0 & 0 & 0 & 47 & 1 \\ \hline
        \tket{
    \draw[thick, blue] (0.0,0) -- (0.5, 0);
    \draw[thick, blue] (1.0,0) -- (1.5, 0);
    \node[circle, fill=blue, inner sep=1.5pt] at (0,0) {};
    \node[circle, fill=blue, inner sep=1.5pt] at (0.5,0) {};
    \node[circle, fill=blue, inner sep=1.5pt] at (1,0) {};
    \node[circle, fill=blue, inner sep=1.5pt] at (1.5,0) {};
    \node[circle, fill=mydarkgreen, inner sep=1.5pt] at (2,0) {};
    \node[circle, fill=red, inner sep=1.5pt] at (2.5,0) {};
}
        & 0 & -1 & 1 & 0 & 0 & 0 & -1 & 0 & 0 & 0 & 47 & 1 \\ \hline
                    \tket{
    \draw[thick, blue] (0.0,0) -- (0.5, 0);
    \draw[thick, blue] (1.0,0) -- (1.5, 0);
    \node[circle, fill=blue, inner sep=1.5pt] at (0,0) {};
    \node[circle, fill=blue, inner sep=1.5pt] at (0.5,0) {};
    \node[circle, fill=blue, inner sep=1.5pt] at (1,0) {};
    \node[circle, fill=blue, inner sep=1.5pt] at (1.5,0) {};
    \node[circle, fill=red, inner sep=1.5pt] at (2,0) {};
    \node[circle, fill=blue, inner sep=1.5pt] at (2.5,0) {};
}
        & 1 & 0 & -1 & 0 & 0 & 0 & 0 & -1 & 0 & 0 & 47 & 1 \\ \hline
    \tket{
    \draw[thick, blue] (0.0,0) -- (0.5, 0);
        \draw[thick, blue] (1.0,0) -- (1.5, 0);
    \node[circle, fill=blue, inner sep=1.5pt] at (0,0) {};
    \node[circle, fill=blue, inner sep=1.5pt] at (0.5,0) {};
    \node[circle, fill=blue, inner sep=1.5pt] at (1,0) {};
    \node[circle, fill=blue, inner sep=1.5pt] at (1.5,0) {};
    \node[circle, fill=red, inner sep=1.5pt] at (2,0) {};
    \node[circle, fill=mydarkgreen, inner sep=1.5pt] at (2.5,0) {};
}
        & 0 & 1 & -1 & 0 & 0 & 0 & 0 & 0 & -1 & 0 & 47 & 1 \\ \hline
                    \tket{
    \draw[thick, blue] (0.0,0) -- (0.5, 0);
    \node[circle, fill=blue, inner sep=1.5pt] at (0,0) {};
    \node[circle, fill=blue, inner sep=1.5pt] at (0.5,0) {};
    \node[circle, fill=blue, inner sep=1.5pt] at (1,0) {};
    \node[circle, fill=mydarkgreen , inner sep=1.5pt] at (1.5,0) {};
    \node[circle, fill=red, inner sep=1.5pt] at (2,0) {};
    \node[circle, fill=mydarkgreen, inner sep=1.5pt] at (2.5,0) {};
}
        & -1 & 2 & -1 & -2 & 1 & 0 & -1 & 0 & -1 & 1 & 11 & 1 \\ \hline
                    \tket{
    \draw[thick, blue] (0.0,0) -- (0.5, 0);
    \node[circle, fill=blue, inner sep=1.5pt] at (0,0) {};
    \node[circle, fill=blue, inner sep=1.5pt] at (0.5,0) {};
    \node[circle, fill=red, inner sep=1.5pt] at (1,0) {};
    \node[circle, fill=mydarkgreen , inner sep=1.5pt] at (1.5,0) {};
    \node[circle, fill=red, inner sep=1.5pt] at (2,0) {};
    \node[circle, fill=mydarkgreen, inner sep=1.5pt] at (2.5,0) {};
}
        & 0 & 2 & -2 & 0 & 0 & 0 & -1 & 0 & -3 & 0 & 11 & 1 \\ \hline
    \vdots & \vdots & \vdots & \vdots & \vdots & \vdots & \vdots & \vdots & \vdots & \vdots & \vdots & \vdots & \vdots 
    \\\hline
\end{array}
\end{equation}
Here, we do not list the full fragmentation structure for brevity. As shown in the table, different irreps---each represented by one of their their corresponding basis states---have at least one commutant that acts differently. For larger systems sizes $L$, the number of irreps, along with their corresponding $D_\lambda$ and $d_\lambda$, admit closed-form expressions. In particular, the number of irreps corresponding to basis vector states containing $p$ singular dots is given by
\begin{equation}
\begin{cases}
   n_{p}  =1& p = 1\\
   n_{p}  = m(m-1)^{p -1}& p > 1
\end{cases}.
\end{equation}
Each such irrep has degeneracy $d_\lambda=1$, and its dimension $D_{L, p}$ can be extracted from the generating function:
\begin{equation}
G(z)\equiv \sum_{\ell = 0}^{\infty}D_{\ell, p}z^{\ell}, \quad D_{0, 0} \equiv 1, \quad  D_{\ell, p} \equiv 0 \textrm{ if $p > \ell$},
\end{equation}
where
\begin{equation}
G(z) = \frac{2(m-1)}{m - 2 + m\sqrt{1 - 4(m -1)z^2}}\left(\frac{1-\sqrt{1 - 4(m -1)z^2}}{2(m-1)z}\right)^p.
\end{equation}

The dimensions of the irreps from the CA approach match exactly with the degree of the corresponding polynomial factors in the ICPF approach. Moreover, the basis of each CA Krylov subspace is annihilated by its associated characteristic polynomial evaluated on the Hamiltonian. For example, \(\phi_{1}^{(87)}(H)
        \tket{
    \draw[thick, blue] (0.0,0) -- (0.5, 0);
    \draw[thick, blue] (1.0,0) -- (1.5, 0);
    \draw[thick, blue] (2.0,0) -- (2.5, 0);
    \node[circle, fill=blue, inner sep=1.5pt] at (0,0) {};
    \node[circle, fill=blue, inner sep=1.5pt] at (0.5,0) {};
    \node[circle, fill=blue, inner sep=1.5pt] at (1,0) {};
    \node[circle, fill=blue, inner sep=1.5pt] at (1.5,0) {};
    \node[circle, fill=blue, inner sep=1.5pt] at (2,0) {};
    \node[circle, fill=blue, inner sep=1.5pt] at (2.5,0) {};
} 
= 0\). Consequently, we conclude that the ICPF approach is equivalent to the CA approach in the $L=6$ PF model. This equivalence has been verified to hold at least up to $L \le 8$.

\section{Temperley-Lieb Model}

\subsection{Model description}

We now study the Temperley-Lieb (TL) model~\cite{PhysRevB.40.4621, Aufgebauer_2010, M.T.Batchelor_1990, Read2007ti}, which belongs to the family of the pair-flip models, but exhibits \emph{quantum fragmentation} under the CA approach. The model is defined on a $m$-level chain under OBC and is decribed by the Hamiltonian:
\begin{equation}
H 
= \sum_{i=1}^{L-1} J_{i, i+1} e_{i}
= \sum_{i=1}^{L-1} J_{i, i+1}
\sum_{\alpha, \beta=1}^{m} 
\left( \ket{\alpha \alpha} \lket{\beta \beta} \right)_{i, i+1}
,\end{equation}
where $e_i = \left( \ket{\alpha \alpha} \lket{\beta \beta} \right)_{i, i+1}$. Similar to the pair-flip model, the TL model also has $m$ hermitian conserved charges:
\begin{equation}
\label{eqnSM: conserved charges of TL}
N^{\alpha} = \sum_{j}(-1)^{j}N^{\alpha}_{j}
,\end{equation}
where $\alpha$ runs from $1$ to $m$.

The basis states of the Krylov subspace in the TL model admit a convenient graphical representation. In this notation, each pair of maximally entangled sites is represented by a dimer:
$$
    \ket{\psi_{\textrm{dimer}}}
    =
    \tket{
    \draw[thick] (0.0,0) -- (0.5, 0);
    \node[circle, fill, inner sep=1.5pt] at (0,0) {};
    \node[circle, fill, inner sep=1.5pt] at (0.5,0) {};}
    =\sum_{\alpha=1}^{m}\ket{\alpha\alpha}
,$$
while frozen states are depicted as isolated dots, which are annihilated by the action of $e_i$. The advantage of this graphical representation is that it makes the action of $e_i$ more transparent. Recall that $e_i$ projects the state on sites $i$ and $i+1$ onto a maximally entangled pair---graphically, this corresponds to inserting a dimer between sites. The action of $e_i$ falls into one of four cases:
\begin{enumerate}
\item
If sites $i$ and $i+1$ are already connected by a dimer, then $e_i$ leaves the configuration unchanged and simply multiplies the state by a factor of $m$:
$$
e_i \ket{\;\tikz[baseline={([yshift=1ex]current bounding box.center)}]{
    \draw[thick ] (0,0) -- (0.5, 0);
    \node[circle, fill, inner sep=1.5pt, label=below:\(\scriptstyle i \)] at (0,0) {};
    \node[circle, fill, inner sep=1.5pt,  label=below:\( \scriptstyle i+1 \)] at (0.5,0) {};}} = m 
\ket{\;\tikz[baseline={([yshift=1ex]current bounding box.center)}]{
    \draw[thick ] (0,0) -- (0.5, 0);
    \node[circle, fill, inner sep=1.5pt, label=below:\(\scriptstyle i \)] at (0,0) {};
    \node[circle, fill, inner sep=1.5pt,  label=below:\( \scriptstyle i+1 \)] at (0.5,0) {};}}
.$$
\item 
If site $i$ is an isolated dot while site $i+1$ is connected by a dimer to some other site $l$, then the action of $e_i$ breaks the dimer between $i+1$ and $l$, leaving $l$ as an isolated dot, and forms a new dimer between $i$ and $i+1$:
$$
\ket{\;\tikz[baseline={([yshift=1ex]current bounding box.center)}]{
    \draw[thick ] (0.5,0) -- (1.0, 0);
    \node[circle, fill, inner sep=1.5pt, label=below:\(\scriptstyle i \)] at (0,0) {};
    \node[circle, fill, inner sep=1.5pt,  label=below:\( \scriptstyle i+1\)] at (0.5,0) {};
    \node[circle, fill, inner sep=1.5pt,  label=below:\( \scriptstyle l \)] at (1.0,0) {};
    }} = 
    \ket{\;\tikz[baseline={([yshift=1ex]current bounding box.center)}]{
    \draw[thick ] (0,0) -- (0.5, 0);
    \node[circle, fill, inner sep=1.5pt, label=below:\(\scriptstyle i \)] at (0,0) {};
    \node[circle, fill, inner sep=1.5pt,  label=below:\( \scriptstyle i +1 \)] at (0.5,0) {};
    \node[circle, fill, inner sep=1.5pt,  label=below:\( \scriptstyle l \)] at (1.0,0) {};
    }}
.$$
\item 
If sites $i$ and $i+1$ are each connected to different sites---$i$ to $k$ and $i+1$ to $l$---then the action of $e_i$ removes both existing dimers and replaces them with two new ones: one connecting $i$ and $i+1$, and the other connecting $k$ and $l$:
$$
e_i\ket{\;\tikz[baseline={([yshift=1ex]current bounding box.center)}]{
    \draw[thick ] (0,0) -- (0.5, 0);
    \draw[thick ] (1.0,0) -- (1.5, 0);
    \node[circle, fill, inner sep=1.5pt, label=below:\(\scriptstyle k \)] at (0,0) {};
    \node[circle, fill, inner sep=1.5pt,  label=below:\( \scriptstyle i \)] at (0.5,0) {};
    \node[circle, fill, inner sep=1.5pt,  label=below:\( \scriptstyle i+1  \)] at (1.0,0) {};
    \node[circle, fill, inner sep=1.5pt,  label=below:\( \scriptstyle l \)] at (1.5,0) {};
    }} 
    = 
    \ket{\;\tikz[baseline={([yshift=0ex]current bounding box.center)}]{
    \draw[thick ] (0.5,0) -- (1, 0);
    \draw[thick] (A) to[out=30, in=150] (D); 
    \node[circle, fill, inner sep=1.5pt, label=below:\(\scriptstyle k \)] (A) at (0,0) {};
    \node[circle, fill, inner sep=1.5pt,  label=below:\( \scriptstyle i \)] (B) at (0.5,0) {};
    \node[circle, fill, inner sep=1.5pt,  label=below:\( \scriptstyle i+1  \)] (C) at (1.0,0) {};
    \node[circle, fill, inner sep=1.5pt,  label=below:\( \scriptstyle l \)] (D) at (1.5,0) {};
    }}
.$$
\item
If sites $i$ and $i+1$ are isolated, corresponding to a frozen state, then $e_i$ anniliates the state:
$$
e_i \ket{\;\tikz[baseline={([yshift=1ex]current bounding box.center)}]{
    \node[circle, fill, inner sep=1.5pt, label=below:\(\scriptstyle i \)] at (0,0) {};
    \node[circle, fill, inner sep=1.5pt,  label=below:\( \scriptstyle i+1 \)] at (0.5,0) {};}} = 0
.$$
\end{enumerate}
The frozen state of size $l$ starting from site $i_s$ satisfies
\begin{equation}\label{eqnSM: Frozen states defintion}
e_i\ket{\psi_{\textrm{frozen}} } =0, \quad i = i_s, \cdots, i_s+l-2
,\end{equation}
which can be obtained by solving the ansatz
\begin{equation}
 \ket{\psi_{\textrm{frozen}}} = \sum_{\alpha_1, \cdots \alpha_L}C_{\alpha_1\cdots \alpha_L}\ket{\alpha_1\cdots \alpha_L}.
\end{equation}
Substituting the ansatz into equation \eqref{eqnSM: Frozen states defintion} yields a system of linear equations for the coefficients $C_{\alpha_1\cdots \alpha_n}$, all with rational coefficients. Consequently, $C_{\alpha_1\cdots \alpha_n}$ admits rational coefficients.

\subsection{Hilbert space fragmentation}

In addition to the conserved charges in \eqref{eqnSM: conserved charges of TL}, the model features more conserved charges of the form
\begin{equation}
M_{\alpha}^{\beta} = \sum_{j} (M_j)_\alpha^\beta, \;\;\; 
\begin{cases}
    (M_j)_\alpha^\beta = -(\ket{\beta}\lket{\alpha})_j & j \textrm{ even}
    \\
    (M_j)_\alpha^\beta = (\ket{\alpha}\lket{\beta})_j & j \textrm{ odd}
.\end{cases} 
\end{equation}
The conserved charges $N^{\alpha}$ can be expressed in terms of $M_{\alpha}^{\beta}$ by noting that $(M_j)_{\alpha}^{\alpha} \equiv (-1)^{j}N^{\alpha}_j$. The operators $M_{\alpha}^{\beta}$ ($1 \leq \alpha, \beta \leq m$) then generate an algebra of $m^2 - 1$ independent elements, corresponding to the global $U(m)$ symmetry of the model. To obtain the most general commutants, we follow the construction in \cite{Read2007ti} and define:
\begin{equation}
\overline{M}_{a_{1}\cdots a_{k}}{}^{b_{1}\cdots b_{k}} = \sum_{i_1 < \cdots < i_k} \prod_{l =1}^{k}(M_{i_l})_{\alpha_{l}}^{\beta_{l}}, \quad k = 1, \cdots \floor{L/2}
.
\end{equation}
The commutants are then obtained by taking traceless projections of these operators:
\begin{equation} 
M_{a_{1}\cdots a_{k}}{}^{b_{1}\cdots b_{k}} = (P \cdot \overline{M} \cdot P)_{a_{1}\cdots a_{k}}{}^{b_{1}\cdots b_{k}}, 
\end{equation}
where $P$ denotes the Jones–Wenzl projector. These projected operators can be constructed recursively for each $k$. For example, when $k=2$, the projected commutant takes the form:
\begin{equation}
M_{ab}{}^{cd}  = \overline{M}_{ab}{}^{cd}  - \frac{1}{m}\overline{M}_{a{}'b}{}^{ca{}'}\delta^{a}{}_{d} -\frac{1}{m}\overline{M}_{ab{}'}{}^{b{}'d}\delta^{b}{}_{c}  + \frac{1}{m^2}\overline{M}_{a{}'b{}'}{}^{b{}'a{}'}\delta^{b}{}_{c} \delta^{a}{}_{d} .
\end{equation}
Given the commutants, the dimensions $D_{\lambda}$ and the degeneracies $d_{\lambda}$ of each Krylov subspace can be explicitly computed:
\begin{equation}
\label{eqnSM: D and d of TL}
D_{\lambda} = 
\begin{pmatrix}
    L
    \\
    L/2+\lambda
\end{pmatrix}
-
\begin{pmatrix}
    L
    \\
    L/2+\lambda+1
\end{pmatrix}
\;\;,\;\;\;
d_{\lambda} = \left[ 2\lambda+1 \right]_{q}
,\end{equation}
Here, $\lambda =0, 1, \cdots L/2 $ when $L$ is even, and $\lambda = 1/2, \cdots (L-1)/2$ when $L$ is odd. The quantities $\left[x\right]_q$ denote the $q$-deformed integers, defined as 
\begin{equation}
    \left[x\right]_q = \frac{q^{x} - q^{-x}}{ q - q^{-1}}
,\end{equation}
where $q = \frac{m+\sqrt{m^2-4}}{2}$.

For model with $L=4$ and $m=3$, the dimension of the Hilbert space is $3^4=81$. From the CA approach, and according to \eqref{eqnSM: D and d of TL}, the dimension and degeneracies of the Krylov subspaces are given by $(D_\lambda, d_\lambda) = (2, 1), (3, 8), (1, 55)$. Each Krylov subspace can be generated from its corresponding root state:
\begin{equation}
\label{eqnSM: root states of TL}
\begin{cases}
    \lambda = 0, D_\lambda=2, d_\lambda=1:\quad
    \tket{
    \draw[thick ] (0,0) -- (0.5, 0);
    \draw[thick ] (1.0,0) -- (1.5, 0);
    \node[circle, fill, inner sep=1.5pt] at (0,0) {};
    \node[circle, fill, inner sep=1.5pt] at (0.5,0) {};
    \node[circle, fill, inner sep=1.5pt] at (1.0,0) {};
    \node[circle, fill, inner sep=1.5pt] at (1.5,0) {};
    }
    \\
    \lambda = 1, D_\lambda=3, d_\lambda=8:\quad
        \tket{
    \draw[thick ] (0,0) -- (0.5, 0);
    \node[circle, fill, inner sep=1.5pt] at (0,0) {};
    \node[circle, fill, inner sep=1.5pt] at (0.5,0) {};
    \node[circle, fill, inner sep=1.5pt] at (1.0,0) {};
    \node[circle, fill, inner sep=1.5pt] at (1.5,0) {};
    }
    \\
    \lambda = 2, D_\lambda=1, d_\lambda=55:\quad
        \tket{
    \node[circle, fill, inner sep=1.5pt] at (0,0) {};
    \node[circle, fill, inner sep=1.5pt] at (0.5,0) {};
    \node[circle, fill, inner sep=1.5pt] at (1.0,0) {};
    \node[circle, fill, inner sep=1.5pt] at (1.5,0) {};
    }
\end{cases}
\end{equation}

In the ICPF approach, we substitute integer coefficients into the Hamiltonian: $(J_{1, 2}, J_{2, 3}, J_{3, 4}) = (7, 6, -3)$, and
the characteristic polynomial factors into
\begin{equation}
\begin{aligned}
    &\det(\mu \mathds{1} - H) 
    = [ \phi_1^{(1)}(\mu) ]^{55} \cdot \phi_2^{(2)}(\mu) \cdot [\phi_{3}^{(3)}(\mu)]^{8}
    =
    \mu^{55} (\mu^2 - 30 \mu + 192) (\mu^3 - 
   30 \mu^2 + 3 \mu + 2646)^{8}
.\end{aligned}
\end{equation}
This indicates that the Hilbert space decomposes into a 55-fold degenerate 1-dimensional Krylov subspace, a single 2-dimensional Krylov subspace, and an 8-fold degenerate 3-dimensional Krylov subspace, which matches the predictions from the CA approach. Moreover, from \eqref{eqnSM: root states of TL}, all root states are rational in the Hamiltonian basis, which leads to the conclusion ICPF $=$ CA for $ L = 4$. We explicitly verified that ICPF $=$ CA for $L \leq 8$.

For arbitrary $L$, ICFP $\succeq$ CA can be established by showing that the root states of the Krylov subspace have rational coefficients when expressed in the Hamiltonian basis. These root states, which correspond to the irreducible representations of the bond algebra in the TL model, are explicitly known \cite{aufgebauer2010quantum}. More explicitly, for an irreducible representation labeled by $\lambda$, with $(D_\lambda, d_\lambda)$ given in \eqref{eqnSM: D and d of TL}, the corresponding root states are given by
\begin{equation} \label{eqnSM: TL root states}
  \ket{\psi_{\textrm{dimer}}}_{1, 2}\otimes \ket{\psi_{\textrm{dimer}}}_{3, 4}\cdots \otimes \ket{\psi_{\textrm{dimer}}}_{L-2\lambda -1, L-2\lambda}\otimes \ket{\psi_{\textrm{frozen}}}
,\end{equation}
where $\ket{\psi_{\textrm{frozen}}}$ represents a frozen state for TL model of size $L - 2\lambda$, satisfying \eqref{eqnSM: Frozen states defintion} with $i_s=2\lambda+1$. The root states can be represented graphically as
\begin{equation}
     \tket{
    \draw[thick] (0.0,0) -- (0.5, 0);
    \draw[thick] (1.0,0) -- (1.5, 0);
    \draw[thick] (2.5,0) -- (3.0, 0);
    \node[circle, fill, inner sep=1.5pt] at (0,0) {};
    \node[circle, fill, inner sep=1.5pt] at (0.5,0) {};
    \node[circle, fill, inner sep=1.5pt] at (1,0) {};
    \node[circle, fill, inner sep=1.5pt] at (1.5,0) {};
    \node[circle, fill, inner sep=1.5pt] at (2.5,0) {};
    \node[circle, fill, inner sep=0.3pt] at (1.85,0) {};
    \node[circle, fill, inner sep=0.3pt] at (2.0,0) {};
    \node[circle, fill, inner sep=0.3pt] at (2.15,0) {};
    \node[circle, fill, inner sep=1.5pt] at (3.0,0) {};
    \node[circle, fill, inner sep=1.5pt] at (3.5,0) {};
    \node[circle, fill, inner sep=1.5pt] at (4.0,0) {};
    \node[circle, fill, inner sep=1.5pt] at (5.0,0) {};
    \node[circle, fill, inner sep=0.3pt] at (4.35,0) {};
    \node[circle, fill, inner sep=0.3pt] at (4.5,0) {};     
    \node[circle, fill, inner sep=0.3pt] at (4.65,0) {};
    }.
\end{equation}
There are exactly $d_{\lambda}$ independent frozen states $\ket{\psi_{\text{frozen}}}$, each corresponding to a distinct root state in a $d_\lambda$-fold degenerate Krylov subspace. Since both $\ket{\psi_\text{dimer}}$ and $\ket{\psi_{\text{frozen}}}$ have rational coefficients in the Hamiltonian basis, the root states also admit rational coefficients in the same basis.

\section{The Quantum Breakdown Model}
The (bosonic) quantum breakdown model describes a 1D bosonic chain with size $L$, each site consist of 2 hardcore bosonic modes $b_{i,\nu}$ labeled by $\nu=1,2$, which satisfy the commutation relations 
\begin{equation}
b_{i,\nu}b_{j,\nu'}^\dagger-(1-2\delta_{ij}\delta_{\nu\nu'})b_{j,\nu'}^\dagger b_{i,\nu}=\delta_{ij}\delta_{\nu\nu'}\ .
\end{equation}
The dynamics is governed by the Hamiltonian:
\begin{equation}
    H = 
    \sum_{i=1}^{L-1} \sum_{\nu=1,2} J_{i}^{\nu} (b_{i+1,1}^{\dagger} b_{i+1,2}^{\dagger} b_{i,\nu} 
    + \text{h.c.})
    +
    \sum_{i=1}^{L} \mu_{i} \sum_{\nu=1,2} \hat{n}_{i,\nu}
,\end{equation}
where h.c. stands for Hermitian conjugate, and $\hat{n}_{i,\nu}= b_{i,\nu}^{\dagger} b_{i,\nu}$ is the number operator of mode $\nu$ at site $i$. The potential term in the Hamiltonian counts the total number of bosons at each site; therefore, the commutants of the model conserves the number of bosons at each site. The basis of the model can be visualized as a $2\times 3$ grid of circles, where the three columns correspond to lattice sites $i=1, 2, 3$, and the two rows correspond to different modes. Filled circles ($\bullet$) represent occupied states, while empty circles ($\circ$) represent unoccupied states. For example, $
\begin{smallmatrix}
\circ & \bullet & \circ & \circ \\
\circ & \circ & \circ & \circ \\
\end{smallmatrix}
$ represents a basis with $L=4$, in which the $\nu=1$ mode at site $i=2$ is occupied.

The bond algebra $\mathcal{A}$ is generated by the local terms appearing in the Hamiltonian,
\begin{equation}
    \mathcal{A}
    = \mathrm{gen}\!\left\{
    b_{i+1,1}^{\dagger} b_{i+1,2}^{\dagger} b_{i,\nu}
    + \mathrm{h.c.},\;
    \hat n_{i,1} + \hat n_{i,2}
    \;\middle|\;
    i,\nu
    \right\}.
\end{equation}
The commutants of the model can then be written in terms of the following operators:
\begin{enumerate}
    \item The identity $\mathbb{I}_{i}$
    \item The number operator: $\hat{n}_{i, \nu}$
    \item The hopping operators at each site: $\hat{o}_{i, 1}=b_{i,1} b_{i,2}^{\dagger}$, 
    $\hat{o}_{i, 2}=b_{i,2} b_{i,1}^{\dagger}$
    \item The product of number operators: $\hat{n}_{i}^{(1, 2)} = \hat{n}_{i, 1} \hat{n}_{i, 2}$
\end{enumerate}
Focusing on $L=3$, the commutant algebra $\mathcal{C}$ of the model is generated by arbitrary linear combinations and multiplications of the following generators:
\begin{enumerate}
    \item $\mathbb{I}_{1} \mathbb{I}_{2} \mathbb{I}_{3}$, 
    $\mathbb{I}_{1} \mathbb{I}_{2} \hat{o}_{3, \nu_3}$, and 
    $\mathbb{I}_{1} \hat{o}_{2, \nu_2} \hat{o}_{3, \nu_3}$
    
    \item $\hat{n}_{1}^{(1, 2)} \hat{n}_{2}^{(1, 2)} \hat{n}_{3}^{(1, 2)}$, 
    $\hat{n}_{1}^{(1, 2)} \hat{n}_{2}^{(1, 2)} \hat{o}_{3, \nu_3}$, and 
    $\hat{n}_{1}^{(1, 2)} \hat{o}_{2, \nu_2} \hat{o}_{3, \nu_3}$
    
    \item $\hat{o}_{1, \nu_1} \hat{o}_{2, \nu_2} \hat{o}_{3, \nu_3}$
    
    \item $\hat{n}_{1}^{(1, 2)} \hat{n}_{2}^{(1, 2)} \hat{n}_{3, \nu_3}$, 
    $\hat{n}_{1}^{(1, 2)} \hat{n}_{2, \nu_2} \hat{o}_{3, \nu_3}$, 
    and $\hat{n}_{1, \nu_1} \hat{o}_{2, \nu_2} \hat{o}_{3, \nu_3}$
    
    \item $\mathbb{I}_{1} \left(\hat{n}_{2,1}-\hat{n}_{2,2}\right) \hat{o}_{3, \nu_3}$ 
    and 
    $\mathbb{I}_{1} \left(\hat{n}_{2,1}-\hat{n}_{2}^{(1, 2)}\right) \hat{o}_{3, \nu_3}$

    \item $\hat{n}_{1, \nu_1} \left(\hat{n}_{2,1}-\hat{n}_{2,2}\right) \hat{o}_{3, \nu_3}$ 
    and 
    $\hat{n}_{1, \nu_1} \left(\hat{n}_{2,1}-\hat{n}_{2}^{(1, 2)}\right) \hat{o}_{3, \nu_3}$

    \item $\hat{o}_{1, \nu_1} \left(\hat{n}_{2,1}-\hat{n}_{2,2}\right) \hat{o}_{3, \nu_3}$ 
    and 
    $\hat{o}_{1, \nu_1} \left(\hat{n}_{2,1}-\hat{n}_{2}^{(1, 2)}\right) \hat{o}_{3, \nu_3}$

    \item $\left[ \mathbb{I}_{1} \hat{n}_{2, 1} 
    + \hat{n}_{1, 1} \mathbb{I}_{2}
    + \hat{n}_{1, 2} \mathbb{I}_{2}\right]
    \hat{o}_{3, \nu_3}$

    \item $\left[ \hat{n}_{1, 1} \hat{n}_{2, 1} 
    + \hat{n}_{1, 2} \hat{n}_{2, 1} 
    + \hat{n}_{1}^{(1, 2)} \mathbb{I}_{2}\right]
    \hat{o}_{3, \nu_3}$

    \item $\hat{n}_{1}^{(1, 2)}
    \left[ \hat{n}_{2, 1} \hat{n}_{3, 1} 
    + \hat{n}_{2, 2} \hat{n}_{3, 1} 
    + \hat{n}_{2}^{(1, 2)} \mathbb{I}_{3}\right]$

    \item $2 \hat{n}_{1, 1} \mathbb{I}_{2} \mathbb{I}_{3}
    + 2 \hat{n}_{1, 2} \mathbb{I}_{2} \mathbb{I}_{3}
    + \mathbb{I}_{1} \hat{n}_{2, 1} \mathbb{I}_{3}
    + \mathbb{I}_{1} \hat{n}_{2, 2} \mathbb{I}_{3}
    + \mathbb{I}_{1} \mathbb{I}_{2} \hat{n}_{3, 1}$

    \item $\hat{n}_{1, 1} \mathbb{I}_{2} \mathbb{I}_{3}
    + \hat{n}_{1, 2} \mathbb{I}_{2} \mathbb{I}_{3}
    + \mathbb{I}_{1} \hat{n}_{2}^{(1, 2)} \mathbb{I}_{3}
    + \mathbb{I}_{1} \hat{n}_{2, 1} \hat{n}_{3, 1}
    + \mathbb{I}_{1} \hat{n}_{2, 2} \hat{n}_{3, 1}
    - 2 \mathbb{I}_{1} \hat{n}_{2}^{(1, 2)} \hat{n}_{3, 1}$
\end{enumerate}
The commutants listed above decompose the Hilbert space into CA subspaces that are dynamically disconnected from each other:
\begin{equation}
\label{eqnSM: CA Krylov subspaces QBM}
\begin{aligned}
\mathcal{K}_{\text{CA}, 1} &= 
\text{span}
\left(
    \boxed{\begin{smallmatrix}
    \circ & \bullet & \circ \\
    \circ & \circ & \circ
    \end{smallmatrix}}
    \;,\;
    \boxed{\begin{smallmatrix}
    \circ & \circ & \circ \\
    \circ & \bullet & \circ
    \end{smallmatrix}}
    \;,\;
    \boxed{\begin{smallmatrix}
    \circ & \circ & \bullet \\
    \circ & \circ & \bullet
    \end{smallmatrix}}
\right)
,
\\
\mathcal{K}_{\text{CA}, 2} &= 
\text{span}
\left(
    \boxed{\begin{smallmatrix}
    \bullet & \bullet & \circ \\
    \bullet & \bullet & \circ
    \end{smallmatrix}}
    \;,\;
    \boxed{\begin{smallmatrix}
    \bullet & \bullet & \bullet \\
    \bullet & \circ & \bullet
    \end{smallmatrix}}
    \;,\;
    \boxed{\begin{smallmatrix}
    \bullet & \circ & \bullet \\
    \bullet & \bullet & \bullet
    \end{smallmatrix}}
\right)
,
\\
\mathcal{K}_{\text{CA}, 3} &= 
\text{span}
\left(
    \boxed{\begin{smallmatrix}
    \bullet & \circ & \circ \\
    \circ & \circ & \circ
    \end{smallmatrix}}
    \;,\;
    \boxed{\begin{smallmatrix}
    \circ & \circ & \circ \\
    \bullet & \circ & \circ
    \end{smallmatrix}}
    \;,\;
    \boxed{\begin{smallmatrix}
    \circ & \bullet & \circ \\
    \circ & \bullet & \circ
    \end{smallmatrix}}
    \;,\;
    \boxed{\begin{smallmatrix}
    \circ & \circ & \bullet \\
    \circ & \bullet & \bullet
    \end{smallmatrix}}
    \;,\;
    \boxed{\begin{smallmatrix}
    \circ & \bullet & \bullet \\
    \circ & \circ & \bullet
    \end{smallmatrix}}
\right)
,
\\
\mathcal{K}_{\text{CA}, 4} &= 
\text{span}
\left(
    \boxed{\begin{smallmatrix}
    \bullet & \bullet & \circ \\
    \bullet & \circ & \circ
    \end{smallmatrix}}
    \;,\;
    \boxed{\begin{smallmatrix}
    \bullet & \circ & \circ \\
    \bullet & \bullet & \circ
    \end{smallmatrix}}
    \;,\;
    \boxed{\begin{smallmatrix}
    \bullet & \circ & \bullet \\
    \bullet & \circ & \bullet
    \end{smallmatrix}}
    \;,\;
    \boxed{\begin{smallmatrix}
    \circ & \bullet & \bullet \\
    \bullet & \bullet & \bullet
    \end{smallmatrix}}
       \;,\;
    \boxed{\begin{smallmatrix}
    \bullet & \bullet & \bullet \\
    \circ & \bullet & \bullet
    \end{smallmatrix}}
\right)
,
\\
\mathcal{K}_{\text{CA}, 5} &= 
\text{span}
\left(
    \boxed{\begin{smallmatrix}
    \bullet & \circ & \circ \\
    \bullet & \circ & \circ
    \end{smallmatrix}}
    \;,\;
    \boxed{\begin{smallmatrix}
    \circ & \bullet & \circ \\
    \bullet & \bullet & \circ
    \end{smallmatrix}}
    \;,\;
    \boxed{\begin{smallmatrix}
    \circ & \circ & \bullet \\
    \bullet & \bullet & \bullet
    \end{smallmatrix}}
    \;,\;
    \boxed{\begin{smallmatrix}
    \circ & \bullet & \bullet \\
    \bullet & \circ & \bullet
    \end{smallmatrix}}
    \;,\;
    \boxed{\begin{smallmatrix}
    \bullet & \bullet & \circ \\
    \circ & \bullet & \circ
    \end{smallmatrix}}
    \;,\;
    \boxed{\begin{smallmatrix}
    \bullet & \circ & \bullet \\
    \circ & \bullet & \bullet
    \end{smallmatrix}}
    \;,\;
    \boxed{\begin{smallmatrix}
    \bullet & \bullet & \bullet \\
    \circ & \circ & \bullet
    \end{smallmatrix}}
\right)
,
\\
\mathcal{K}_{\text{CA}, 6} &= 
\text{span}
\left(
    \boxed{\begin{smallmatrix}
    \bullet & \bullet & \circ \\
    \circ & \circ & \circ
    \end{smallmatrix}}
    \;,\;
    \boxed{\begin{smallmatrix}
    \bullet & \circ & \circ \\
    \circ & \bullet & \circ
    \end{smallmatrix}}
    \;,\;
    \boxed{\begin{smallmatrix}
    \bullet & \circ & \bullet \\
    \circ & \circ & \bullet
    \end{smallmatrix}}
    \;,\;
    \boxed{\begin{smallmatrix}
    \circ & \bullet & \bullet \\
    \circ & \bullet & \bullet
    \end{smallmatrix}}
    \;,\;
    \boxed{\begin{smallmatrix}
    \circ & \bullet & \circ \\
    \bullet & \circ & \circ
    \end{smallmatrix}}
    \;,\;
    \boxed{\begin{smallmatrix}
    \circ & \circ & \circ \\
    \bullet & \bullet & \circ
    \end{smallmatrix}}
    \;,\;
    \boxed{\begin{smallmatrix}
    \circ & \circ & \bullet \\
    \bullet & \circ & \bullet
    \end{smallmatrix}}
\right)
,
\\
\mathcal{K}_{\text{CA}, 7;1} &= 
\text{span}
\left(
    \boxed{\begin{smallmatrix}
    \bullet & \circ & \bullet \\
    \circ & \circ & \circ
    \end{smallmatrix}}
    \;,\;
    \boxed{\begin{smallmatrix}
    \circ & \circ & \bullet \\
    \bullet & \circ & \circ
    \end{smallmatrix}}
    \;,\;
    \boxed{\begin{smallmatrix}
    \circ & \bullet & \bullet \\
    \circ & \bullet &\circ
    \end{smallmatrix}}
\right)
,
\\
\mathcal{K}_{\text{CA}, 7;2} &= 
\text{span}
\left(
    \boxed{\begin{smallmatrix}
    \bullet & \circ & \circ \\
    \circ & \circ & \bullet
    \end{smallmatrix}}
    \;,\;
    \boxed{\begin{smallmatrix}
    \circ & \circ &\circ \\
    \bullet & \circ & \bullet
    \end{smallmatrix}}
    \;,\;
    \boxed{\begin{smallmatrix}
    \circ & \bullet & \circ \\
    \circ & \bullet& \bullet
    \end{smallmatrix}}
\right)
,
\\
\mathcal{K}_{\text{CA}, 8;1} &= 
\text{span}
\left(
    \boxed{\begin{smallmatrix}
    \bullet & \circ & \bullet\\
    \bullet & \circ & \circ
    \end{smallmatrix}}
    \;,\;
    \boxed{\begin{smallmatrix}
    \bullet & \bullet &\bullet \\
    \circ & \bullet & \circ
    \end{smallmatrix}}
    \;,\;
    \boxed{\begin{smallmatrix}
    \circ & \bullet & \bullet \\
    \bullet & \bullet& \circ
    \end{smallmatrix}}
\right)
,
\\
\mathcal{K}_{\text{CA}, 8;2} &= 
\text{span}
\left(
    \boxed{\begin{smallmatrix}
    \bullet & \circ & \circ\\
    \bullet & \circ & \bullet
    \end{smallmatrix}}
    \;,\;
    \boxed{\begin{smallmatrix}
    \bullet & \bullet &\circ \\
    \circ& \bullet & \bullet
    \end{smallmatrix}}
    \;,\;
    \boxed{\begin{smallmatrix}
    \circ & \bullet & \circ \\
    \bullet & \bullet& \bullet
    \end{smallmatrix}}
\right)
,
\end{aligned}
\end{equation}
where we list only the Krylov subspaces with dimensions greater than one. The Krylov subspaces $\mathcal{K}_{\text{CA}, 1}$, $\mathcal{K}_{\text{CA}, 3}$, $\mathcal{K}_{\text{CA}, 5}$, $\mathcal{K}_{\text{CA}, 7}$ are related to $\mathcal{K}_{\text{CA}, 2}$, $\mathcal{K}_{\text{CA}, 4}$, $\mathcal{K}_{\text{CA}, 6}$, $\mathcal{K}_{\text{CA}, 8}$, respectively, by particle-hole symmetry. The Krylov subspaces $\mathcal{K}_{\text{CA}, 7}$ and $\mathcal{K}_{\text{CA}, 8}$ are both doubly degenerated. We showed that they are indeed irreducible representations of the bound algebra ${\cal A}$ by verifying that all commutants are proportional to identity once projected onto the Krylov subspaces. By the virtue of Schur's lemma, all ${\cal K}_{\text{CA}}$ listed above are irreducible \cite{grillet2007abstract}. Moreover, the Krylov subspaces listed above are spanned by classical product basis; therefore, the model exhibits classical fragmentation.

In the ICPF approach, all the Krylov subspaces listed in \eqref{eqnSM: CA Krylov subspaces QBM} are further decomposed into smaller ICPF Krylov subspaces. As discussed in our proof in \cref{sectionSM: formal proof of ICPF <= CA} , all smaller ICPF Krylov subspaces depends explicitly on the Hamiltonian coefficients $J_{i}^{\nu}$:
\begin{equation}
\begin{aligned}
    \mathcal{K}_{\text{CA}, 1} =& 
    \text{span}
    \left(
        \boxed{\begin{smallmatrix}
        \bullet & \circ & \circ\\
        \bullet & \circ & \bullet
        \end{smallmatrix}}
        \;,\;
        J_{1}^{1}
        \boxed{\begin{smallmatrix}
        \circ & \bullet & \circ \\
        \bullet & \bullet& \bullet
        \end{smallmatrix}}
        +
        J_{1}^{2}
        \boxed{\begin{smallmatrix}
        \bullet & \bullet &\circ \\
        \circ& \bullet & \bullet
        \end{smallmatrix}}
    \right)
    \oplus
    \text{span}
    \left(
        J_{1}^{2}
        \boxed{\begin{smallmatrix}
        \circ & \bullet & \circ \\
        \bullet & \bullet& \bullet
        \end{smallmatrix}}
        -
        J_{1}^{1}
        \boxed{\begin{smallmatrix}
        \bullet & \bullet &\circ \\
        \circ& \bullet & \bullet
        \end{smallmatrix}}
    \right)
    \\[2ex]
  \mathcal{K}_{\text{CA}, 3} =& 
    \text{span}
    \left(
        \boxed{\begin{smallmatrix}
        \circ & \bullet & \circ\\
        \circ & \bullet & \circ
        \end{smallmatrix}}
        \;,\;
        J_{1}^{1}
        \boxed{\begin{smallmatrix}
        \bullet & \circ & \circ \\
        \circ & \circ& \circ
        \end{smallmatrix}}
        +
        J_{1}^{2}
        \boxed{\begin{smallmatrix}
        \circ & \circ &\circ \\
        \bullet& \circ & \circ
        \end{smallmatrix}}
        \;,\;
        J_{2}^{1}
        \boxed{\begin{smallmatrix}
        \circ & \circ & \bullet \\
        \circ & \bullet& \bullet
        \end{smallmatrix}}
        +
        J_{2}^{2}
        \boxed{\begin{smallmatrix}
        \circ & \bullet &\bullet \\
        \circ& \circ & \bullet
        \end{smallmatrix}}
    \right) \\
         & \oplus \text{span} \left(
        J_{1}^{2}
        \boxed{\begin{smallmatrix}
        \bullet & \circ & \circ \\
        \circ & \circ& \circ
        \end{smallmatrix}}
        -
        J_{1}^{1}
        \boxed{\begin{smallmatrix}
        \circ & \circ &\circ \\
        \bullet & \circ & \circ
        \end{smallmatrix}}
    \right)
    \oplus
    \text{span}
    \left(
        J_{2}^{2}
        \boxed{\begin{smallmatrix}
        \circ & \circ & \bullet \\
        \circ & \bullet& \bullet
        \end{smallmatrix}}
        -
        J_{2}^{1}
        \boxed{\begin{smallmatrix}
        \circ & \bullet &\bullet \\
        \circ & \circ & \bullet
        \end{smallmatrix}}
    \right)
    \\[2ex]
    \mathcal{K}_{\text{CA}, 5} =& 
    \text{span}
    \left(
        \boxed{\begin{smallmatrix}
        \bullet & \circ & \circ\\
        \bullet & \circ & \circ
        \end{smallmatrix}}
        \;,\;
        J_{1}^{1}
        \boxed{\begin{smallmatrix}
        \circ & \bullet & \circ \\
        \bullet & \bullet& \circ
        \end{smallmatrix}}
        +
        J_{1}^{2}
        \boxed{\begin{smallmatrix}
        \bullet & \bullet &\circ \\
        \circ& \bullet & \circ
        \end{smallmatrix}}
        \;,\;
        J_{1}^{1}
        J_{2}^{1}
        \boxed{\begin{smallmatrix}
        \circ & \circ & \bullet \\
        \bullet & \bullet& \bullet
        \end{smallmatrix}}
        +
        J_{1}^{1}
        J_{2}^{2}
        \boxed{\begin{smallmatrix}
        \circ & \bullet &\bullet \\
        \bullet& \circ & \bullet
        \end{smallmatrix}}
        +
        J_{1}^{2}
        J_{2}^{1}
        \boxed{\begin{smallmatrix}
        \bullet & \circ & \bullet \\
        \circ & \bullet&  \bullet
        \end{smallmatrix}}
        +
        J_{1}^{2}
        J_{2}^{2}
        \boxed{\begin{smallmatrix}
        \bullet & \bullet &\bullet \\
        \circ& \circ & \bullet
        \end{smallmatrix}}
    \right) \\[1ex]
        & \oplus \text{span}
        \left( 
        J_{1}^{2}
        \boxed{\begin{smallmatrix}
        \circ & \bullet & \circ \\
        \bullet & \bullet&\circ
        \end{smallmatrix}}
        -
        J_{1}^{1}
        \boxed{\begin{smallmatrix}
        \bullet & \bullet &\circ \\
        \circ& \bullet &   \circ
        \end{smallmatrix}}
        \;,\;
        J_{1}^{2}
        J_{2}^{1}
        \boxed{\begin{smallmatrix}
        \circ & \circ & \bullet \\
        \bullet & \bullet& \bullet
        \end{smallmatrix}}
        +
        J_{1}^{2}
        J_{2}^{2}
        \boxed{\begin{smallmatrix}
        \circ & \bullet &\bullet \\
        \bullet& \circ & \bullet
        \end{smallmatrix}}
        -
        J_{1}^{1}
        J_{2}^{1}
        \boxed{\begin{smallmatrix}
        \bullet & \circ & \bullet \\
        \circ & \bullet&  \bullet
        \end{smallmatrix}}
        -
        J_{1}^{1}
        J_{2}^{2}
        \boxed{\begin{smallmatrix}
        \bullet & \bullet &\bullet\\
        \circ& \circ & \bullet
        \end{smallmatrix}}
    \right) \\[1ex]
         & \oplus \text{span} \left(
        J_{2}^{1}
        \boxed{\begin{smallmatrix}
        \bullet & \bullet &\bullet\\
        \circ& \circ & \bullet
        \end{smallmatrix}}
       -
        J_{2}^{2}
        \boxed{\begin{smallmatrix}
        \bullet & \circ & \bullet \\
        \circ & \bullet&  \bullet
        \end{smallmatrix}}
    \right)
    \oplus
    \text{span}
    \left(
        J_{2}^{1}
        \boxed{\begin{smallmatrix}
        \circ & \circ & \bullet \\
        \bullet & \bullet& \bullet
        \end{smallmatrix}}
        -
        J_{2}^{2}
        \boxed{\begin{smallmatrix}
        \circ & \bullet &\bullet \\
        \bullet& \circ & \bullet
        \end{smallmatrix}}
    \right)
    \\[2ex]
    \mathcal{K}_{\text{CA}, 7;1} =& 
    \text{span}
    \left(
        \boxed{\begin{smallmatrix}
        \circ & \bullet & \bullet\\
        \circ & \bullet & \circ
        \end{smallmatrix}}
        \;,\;
        J_{1}^{1}
        \boxed{\begin{smallmatrix}
        \bullet & \circ & \bullet \\
        \circ & \circ& \circ
        \end{smallmatrix}}
        +
        J_{1}^{2}
        \boxed{\begin{smallmatrix}
        \circ & \circ &\bullet \\
        \bullet& \circ & \circ
        \end{smallmatrix}}
        \right)
        \oplus
        \text{span}
        \left(
        J_{1}^{2}
        \boxed{\begin{smallmatrix}
        \bullet & \circ & \bullet \\
        \circ & \circ& \circ
        \end{smallmatrix}}
        -
        J_{1}^{1}
        \boxed{\begin{smallmatrix}
        \circ & \circ & \bullet\\
        \bullet & \circ & \circ
        \end{smallmatrix}}
        \right)
    \\[2ex]
    \mathcal{K}_{\text{CA}, 7;2} =& 
    \text{span}
    \left(
        \boxed{\begin{smallmatrix}
        \circ & \bullet & \circ\\
        \circ & \bullet & \bullet
        \end{smallmatrix}}
        \;,\;
        J_{1}^{1}
        \boxed{\begin{smallmatrix}
        \bullet & \circ & \circ \\
        \circ & \circ& \bullet
        \end{smallmatrix}}
        +
        J_{1}^{2}
        \boxed{\begin{smallmatrix}
        \circ & \circ &\circ \\
        \bullet& \circ & \bullet
        \end{smallmatrix}}
        \right)
        \oplus
        \text{span}
        \left(
        J_{1}^{2}
        \boxed{\begin{smallmatrix}
        \bullet & \circ & \circ \\
        \circ & \circ& \bullet
        \end{smallmatrix}}
        -
        J_{1}^{1}
        \boxed{\begin{smallmatrix}
        \circ & \circ & \circ\\
        \bullet & \circ & \bullet
        \end{smallmatrix}}
        \right)
    \\[2ex]
  \end{aligned}
\end{equation}

\end{document}